\documentclass[10pt,conference,letterpaper]{IEEEtran}
\usepackage[backend=bibtex,style=ieee,natbib=true]{biblatex} %
\AtEveryBibitem{%
  \clearlist{language}
  \clearlist{location}
  \clearname{editor}
  \clearfield{month}%
  \clearfield{isbn}%
  \clearfield{issn}%
  \clearfield{pages}%
  \clearfield{review}%
  \clearfield{series}%
  \clearlist{address}
  \clearfield{date}
  \clearfield{eprint} 
  \clearfield{doi}

  \ifentrytype{misc}{}{
    \clearfield{url}%
    \clearfield{urldate}%
    \clearfield{urlday}%
    \clearfield{urlmonth}%
    \clearfield{urlyear}%
  }
  
  \ifentrytype{collection}{
    \clearlist{publisher}
    \clearfield{volume}
    \clearfield{number}
  }{}

  \ifentrytype{incollection}{
    \clearlist{publisher}
    \clearfield{volume}
    \clearfield{number}
  }{}

  \ifentrytype{inproceedings}{
    \clearlist{publisher}
    \clearfield{volume}
    \clearfield{number}
  }{}
}

\usepackage[switch]{lineno} 
\newcommand{\myparagraph}[1]{\medskip \noindent \textbf{#1}}

\usepackage{./ressources/ressources}
\usepackage{./ressources/textmacros}
\usepackage{./ressources/my_listings}
\graphicspath{{./ressources/img/}}

\begin{document}
\title{\brelse{}: Efficient Relational Symbolic Execution for Constant-Time at Binary-Level}

\author{
  \IEEEauthorblockN{Lesly-Ann Daniel\IEEEauthorrefmark{1},
    Sébastien Bardin\IEEEauthorrefmark{1},
    Tamara Rezk\IEEEauthorrefmark{2}}\\
  \IEEEauthorblockA{\IEEEauthorrefmark{1}   
    CEA, List, Université Paris-Saclay, France}
  \IEEEauthorblockA{\IEEEauthorrefmark{2}
    INRIA Sophia-Antipolis, INDES Project, France}\\
    lesly-ann.daniel@cea.fr, sebastien.bardin@cea.fr, tamara.rezk@inria.fr
}

\maketitle

\begin{abstract}
  
  The constant-time programming discipline (CT)  is an efficient
  countermeasure against timing side-channel attacks, requiring the
  control flow and the memory accesses to be independent from the
  secrets.
  Yet, writing CT code  is challenging as it demands to reason about pairs of 
  execution traces (2-hypersafety property) and it is generally not
  preserved by the compiler, requiring binary-level analysis. 
  Unfortunately, current verification tools for CT
  either reason at higher level (C or LLVM), or sacrifice bug-finding
  or bounded-verification, or do not scale.
  We tackle the problem of designing an efficient binary-level
  verification tool for CT providing both bug-finding and
  bounded-verification.
  The technique builds on relational symbolic execution enhanced with
  new optimizations dedicated to information flow and binary-level
  analysis, yielding a dramatic improvement \modif{compared to}{over}
  prior work based on symbolic execution. %
  We implement a prototype, \brelse{}, and perform extensive
  experiments on a set of 338 cryptographic implementations,
  demonstrating the benefits of our approach in both bug-finding and
  bounded-verification.
  Using \brelse{}, we also automate %
  a previous manual study of CT preservation by
  compilers. Interestingly, we discovered that \texttt{gcc -O0} and
  backend passes of \texttt{clang} introduce violations of CT in
  implementations that were previously deemed secure by a
  state-of-the-art CT verification tool operating at LLVM level,
  showing the importance of reasoning at binary-level.
\end{abstract}

\section{Introduction}\label{sec:intro} %
Timing channels occur when timing variations in a sequence of events
depends on secret data. They can be exploited by an attacker to
recover secret information such as plaintext data or secret keys. %
\emph{Timing attacks}, unlike other side-channel attacks (e.g based on
power-analysis, electromagnetic radiation or acoustic emanations) do
not require special equipment and can be performed
remotely~\cite{brumley_remote_2005,atluri_remote_2011}. %
First timing attacks exploited secret-dependent \emph{control flow}
with measurable timing differences to recover secret
keys~\cite{kocher_timing_1996} from cryptosystems. %
With the increase of shared architectures (e.g.~infrastructure as a
service) arise more powerful attacks where an attacker can monitor the
cache of the victim and recover information on secret-dependent
\emph{memory
  accesses}~\cite{bernstein_cache-timing_2005,hutchison_cache_2006,percival_cache_2009}.

Therefore, it is of paramount importance to implement adequate
countermeasures to protect cryptographic implementations from these
attacks. %
Simple countermeasures consisting in adding noise or dummy
computations can reduce timing variations and make attacks more
complex. %
Yet, these mitigations eventually become vulnerable to new generations
of attacks and provide only \emph{pseudo}
security~\cite{ronen_pseudo_2018}. %

The \emph{constant-time} programming discipline
(CT)~\cite{barthe_system-level_2014}, a.k.a.~constant-time policy, is
a software-based countermeasure to timing attacks which requires the
control flow and the memory accesses of the program to be independent
from the secret input\footnote{\modif{Also, operands of instructions
    that execute in variable time (e.g., integer division) are
    sometimes required to be independent from secret.}{Some versions
    of CT also require that the size of operands of variable-time
    instructions (e.g.~integer division) is independent from
    secrets.}}. %
Constant-time has been proven to protect against cache-based timing
attacks~\cite{barthe_system-level_2014}, making it the most effective
countermeasure against timing attacks,  %
 already widely used to secure cryptographic implementations
(e.g. BearSSL~\cite{pornin_bearssl_nodate},
NaCL~\cite{hutchison_security_2012},
HACL*~\cite{zinzindohoue_hacl*:_2017}, etc). %

\myparagraph{Problem.} %
Writing constant-time code is complex as it requires low-level
operations deviating from traditional programming
behaviors. %
Moreover, this effort is brittle as it is generally not preserved by
compilers~\cite{simon_what_2018,foresti_when_2016}.  For example,
reasoning about CT requires to know whether the code
\codeinline{c=(x<y)-1} will be compiled to branchless code, but this
depends on the compiler version and
optimization~\cite{simon_what_2018}.
As shown in the attack on a ``constant-time'' implementation of
elliptic curve Curve25519~\cite{foresti_when_2016}, writing
constant-time code is error
prone~\cite{foresti_when_2016,simon_what_2018,ronen_pseudo_2018}.

Several CT-analysis tools have been proposed to analyze source
code~\cite{bacelar_almeida_formal_2013,blazy_verifying_2017}, or LLVM
code~\cite{almeida_verifying_2016,brotzman_casym:_2019},
\modif{counting on the compiler to preserve the property and
  leaving}{but leave} the gap opened for violations introduced in the
executable code either by the compiler~\cite{simon_what_2018} or by
closed-source libraries~\cite{foresti_when_2016}.
Binary-level tools for CT \modif{that use}{using} dynamic
aproaches~\cite{langley_imperialviolet_2010,chattopadhyay_quantifying_2017,
  wang_cached:_2017,wichelmann_microwalk:_2018} can find bugs but miss
vulnerabilities in unexplored portions of the code, making them
incomplete; while static
approaches~\cite{kopf_automatic_2012,doychev_cacheaudit:_2015,doychev_rigorous_2017}
cannot report precise counterexamples -- making them of minor interest
when the implementation cannot be proven secure.  Aside from a
posteriori analysis, correct-by-design
approaches~\cite{almeida_jasmin:_2017,bond_vale:_2017,cauligi_fact:_2017,zinzindohoue_hacl*:_2017}
require to reimplement cryptographic primitives from scratch,
and OS-based
countermeasures~\cite{zhou_software_2016,liu_catalyst:_2016,ge_time_2019,gruss_strong_2017}
incur runtime overhead and require specific OS- or hardware-support.

\myparagraph{Challenges.} %
Two main challenges arise in the verification of constant-time:
\begin{itemize}
\item First, common verification methods do not directly apply because
  information flow properties like CT are not regular safety
  properties but 2-hypersafety
  properties~\cite{clarkson_hyperproperties_2010} (i.e.\ relating two
  execution traces), and their standard reduction to safety on a
  transformed program,
  \emph{self-composition}~\cite{barthe_secure_2004}, is
  inefficient~\cite{terauchi_secure_2005};

\item Second, it is notoriously difficult to adapt formal methods to
  binary-level because of the lack of structure information (data \&
  control) and the explicit representation of the memory as a single
  large array of
  bytes~\cite{DBLP:journals/toplas/BalakrishnanR10,DBLP:conf/fm/DjoudiBG16}.
\end{itemize}

\noindent A technique that scales well on binary code and that
naturally comes into play for bug-finding and bounded-verification is
\emph{symbolic execution}
(SE)~\cite{godefroid_sage:_2012,cadar_symbolic_2013}. While it has
proven very successful for standard safety
properties~\cite{bounimova_billions_2013}, its direct adaptation to CT
and other 2-hypersafety properties through (variants of)
self-composition suffers from a scalability
issue~\cite{balliu_automating_2014,do_exploit_2015,milushev_noninterference_2012}.
Some recent approaches achieve better scaling, but at the cost of
sacrificing either
bounded-verification~\cite{wang_cached:_2017,subramanyan_verifying_2016}
(under-approximation) %
or bug-finding~\cite{brotzman_casym:_2019}
(over-approximations).  %

The idea of \modif{synchronizing two executions}{analyzing pairs of
  executions} for the verification of 2-hypersafety is not new
(e.g.~relational Hoare logic~\cite{benton_simple_2004},
self-composition~\cite{barthe_secure_2004}, product
programs~\cite{butler_relational_2011}, multiple
facets~\cite{austin_flanagan_2012, ngo_etal_2018}). In the context of
symbolic execution, it has first been coined as
\emph{ShadowSE}~\cite{palikareva_shadow_2016} for back-to-back
testing, and later as \emph{relational symbolic execution}
(RelSE)~\cite{farina_relational_2017}.  However, a direct adaptation
of this technique \emph{does not scale in the context of binary-level
  analysis} because of the representation of the memory as a single
large array which prevents sharing between executions, sending \emph{a
  high number of queries} to the constraint solver which could have
been simplified beforehand with a better information flow tracking.

\myparagraph{Proposal.} %
\emph{We tackle the problem of designing an efficient symbolic
  verification tool for constant-time at binary-level that leverages
  the full power of symbolic execution without sacrificing correct
  bug-finding nor bounded-verification.}  We present \brelse{}, the
first efficient binary-level automatic tool for bug-finding and
bounded-verification of constant-time at binary-level.  It is
compiler-agnostic, targets x86 and ARM architectures and does not
require source code.

The technique is based on \emph{relational symbolic
  execution}~\cite{palikareva_shadow_2016,farina_relational_2017}.  It
models two execution traces following the same path in the same
symbolic execution instance %
and \emph{maximizes sharing between them}. We show via experiments
(\cref{sec:scalability}) that \modif{ShadowSE}{RelSE (and ShadowSE)}
alone does not scale at binary-level to analyze CT on real
cryptographic implementations.

Our key technical insights are (1) to use this sharing to track
secret-dependencies and reduce the number of queries sent to the
solver, and (2) to complement it with dedicated optimizations offering
a fine-grained information flow tracking in the memory for efficient
binary analysis.

\brelse{} can analyze about 23 million instructions in 98 min,
(i.e.\ 3860 instructions per second), outperforming similar state of
the art binary-level verification tools based on symbolic
execution~\cite{subramanyan_verifying_2016,wang_cached:_2017}
(cf.~\cref{tab:comparison_se}, page~\pageref{tab:comparison_se}),
while being still correct and complete for CT.

\myparagraph{Contributions.}  Our contributions are the following:
\begin{itemize}

\item We design dedicated optimizations for information flow analysis
  at binary-level. First, we complement relational symbolic execution
  with a new \emph{on-the-fly} simplification for \emph{binary-level}
  analysis, to track secret-dependencies and maximize sharing in the
  memory (\cref{sec:row}).  \modif{and}{Second, we design} new
  simplifications for \emph{information flow} analysis: untainting
  (\cref{sec:untainting}) and fault-packing (\cref{sec:fp}).
  Moreover, we formally prove that our analysis is correct for
  bug-finding and bounded-verification of CT (\cref{sec:proofs}).

\item We propose a verification tool named \brelse{} for CT analysis
  (\cref{sec:implem}).  Extensive experimental evaluation (338
  samples) %
  against standard approaches based on self-composition and
  \modif{standard}{} RelSE (\cref{sec:scalability}) shows that it can
  find bugs in real-world cryptographic implementations much faster
  than these techniques (\(\times 700\) speedup) and can achieve
  bounded-verification when they timeout, \modif{achieving}{with}
  performances close to standard SE (\(\times 1.8\) overhead).

\item In order to prove the effectiveness of \brelse{}, we perform an
  extensive analysis of CT at binary-level. In particular, we analyze
  296 cryptographic implementations previously verified at a
  higher-level (incl.~codes from
  HACL*~\cite{zinzindohoue_hacl*:_2017},
  BearSSL~\cite{pornin_bearssl_nodate},
  NaCL~\cite{hutchison_security_2012}), we replay known bugs in 42
  samples (incl.~Lucky13~\cite{al_fardan_lucky_2013}) %
  and automatically generate counterexamples
  (\cref{sec:effectiveness}).

\item Simon \emph{et al.}~\cite{simon_what_2018} have demonstrated that
  \texttt{clang}'s optimizations break constant-timeness of code. We
  extend this work in four directions -- from 192
  in~\cite{simon_what_2018} to 408 configurations
  (\cref{sec:effectiveness}):
   (1) we automatically analyze the code that was manually checked
    in~\cite{simon_what_2018},
   (2) we add new implementations, 
   (3) we add the \texttt{gcc} compiler and a more recent version of
    \texttt{clang}, 
    (4) we add ARM binaries.
    Interestingly, we discovered that \texttt{gcc -O0} and backend
    passes of \texttt{clang} with \texttt{-O3 -m32 -march=i386}
    introduce violations of CT that cannot be detected by LLVM
    verification tools like ct-verif~\cite{almeida_verifying_2016}.
\end{itemize}

\noindent Our technique is shown to be highly efficient on bug-finding
and bounded-verification compared to alternative approaches, paving
the way to a systematic binary-level analysis of CT on cryptographic
implementations, while our experiments demonstrate the importance of
developing CT-verification tools reasoning at
binary-level. %
Besides CT, the technique can be readily adapted to other
hyperproperties of interest in security (e.g., cache-side channels,
secret-erasure), %
as long as they are restricted to pairs of traces following the same
path.

\section{Background}\label{sec:background} %
In this section, we present the basics of contant-time and symbolic
execution.
Small examples of CT and standard adaptations of symbolic execution
are presented in \cref{sec:motivating}, while a formal definition of
CT is given in \cref{sec:concrete_semantics}.

\subsection{Constant-Time}

Information flow policies regulate the transfer of information between
public and secret domains. To reason about information flow, we
partition the program input into two disjoint sets: \emph{low}
(i.e.\ public) and \emph{high} (i.e.\ secret).
Typical information flow properties require that the observable output
of a program does not depend on the high input
(\textit{non-interference}). CT is a special case requiring both the
program control flow and the memory accesses to be independent from
high input.

Contrary to a standard \emph{safety} property which states that
nothing bad can happen along \emph{one execution trace}, information
flow properties relate \emph{two execution traces} -- they are
\emph{2-hypersafety properties}~\cite{clarkson_hyperproperties_2010}.
Unfortunately, the vast majority of symbolic execution
tools~\cite{godefroid_sage:_2012,cadar_klee:_2008,DBLP:journals/ieeesp/AvgerinosBDGNRW18,DBLP:journals/tocs/ChipounovKC12,shoshitaishvili_sok:_2016,david_binsec/se:_2016}
is designed for safety verification and cannot directly be applied to
2-hypersafety properties.
In principle, 2-hypersafety properties can be reduced to standard
safety properties of a \emph{self-composed
  program}~\cite{barthe_secure_2004} but this reduction alone does not
scale~\cite{terauchi_secure_2005}.

\subsection{Symbolic Execution}
Symbolic Execution
(SE)~\cite{king_symbolic_1976,cadar_symbolic_2013,godefroid_sage:_2012}
consists in executing a program on \emph{symbolic inputs} instead of
concrete input values. Variables and expressions of the program are
represented as terms over symbolic inputs and the current path is
modeled by a \emph{path predicate} (a logical formula), which is the
conjunction of conditional expressions encountered along the
execution.

SE is built upon two main steps. %
(1) \emph{Path search}: at each conditional statement the symbolic
execution \emph{forks}, the expression of the condition is added to the
first branch and its negation to the second branch, then the symbolic
execution continues along satisfiable branches; %
(2) \emph{Constraint solving}: the path predicate can be solved with
an off-the-shelf \emph{automated constraint solver}, typically
SMT~\cite{DBLP:conf/woot/VanegueH12}, to generate a concrete input
exercising the path.

Combining these two steps, SE can explore many different program paths
and generate test inputs exercising these paths. It can also check
local assertions in order to \emph{find bugs} or perform
\emph{bounded-verification} (i.e., verification up to a certain
depth).
Dramatic progresses in program analysis and constraint solving over
the last two decades have made SE a tool of choice for intensive
testing~\cite{bounimova_billions_2013,cadar_symbolic_2013},
vulnerability
analysis~\cite{DBLP:journals/ieeesp/AvgerinosBDGNRW18,DBLP:conf/ndss/AvgerinosCHB11,DBLP:conf/uss/SchwartzAB11}
and other security-related
analysis~\cite{DBLP:conf/sp/YadegariJWD15,bardin_backward-bounded_2017}.

\subsection{Binary-Level Symbolic Execution}\label{sec:bl-se}

Low-level code operates on a set of registers and a single (large)
untyped memory.  During the execution, %
a call stack contains information about the active functions such as
their arguments and local variables. %
A special register \texttt{esp} (stack pointer) indicates the top
address of the call stack and local variables of a function can be
referenced as offsets from the initial
\texttt{esp}\footnote{\texttt{esp} is specific to x86, but this is
  generalizable, e.g.~\texttt{sp} for ARMv7.}.

\myparagraph{Binary-level symbolic execution.} {Binary-level code
  analysis} is notoriously more challenging than source code
analysis~\cite{DBLP:journals/toplas/BalakrishnanR10,DBLP:conf/fm/DjoudiBG16}. First,
evaluation and assignments of source code variables become memory load
and store operations, requiring to reason explicitely about the memory
in a very precise way. Second, the high level control flow structure
(e.g. \texttt{for} loops) is not preserved, \modif{as}{and} there are
dynamic jumps to handle (e.g.~instruction of the form \texttt{jmp
  eax}).

Fortunately, it turns out that SE is less difficult to adapt from
source code to binary code than other semantic analysis -- due to both
the efficiency of SMT solvers and concretization (i.e., simplifying a
formula by constraining some variables to be equal to their observed
runtime values).
Hence, strong binary-level SE tools do exist and have yielded several
highly promising case
studies~\cite{godefroid_sage:_2012,DBLP:journals/ieeesp/AvgerinosBDGNRW18,DBLP:journals/tocs/ChipounovKC12,shoshitaishvili_sok:_2016,david_binsec/se:_2016,bardin_backward-bounded_2017,DBLP:conf/dimva/SalwanBP18}.

\myparagraph{Logical notations.}  Binary-level SE relies on the theory
of bitvectors and arrays, \abv{}~\cite{barrett_smt-lib_2017}. %
Values (e.g.~registers, memory addresses, memory content) are modeled
with fixed-size
bitvectors~\cite{noauthor_fixedsizebitvectors_nodate}. %
We will use the type \(\bvtype{m}\), where \(m\) is a constant number,
to represent symbolic bitvector expressions. %
The memory is modeled with a logical
array~\cite{noauthor_arraysex_nodate,farinier_arrays_2018} of type
\(\memtype{}\) (assuming a $32$ bit architecture).  %
A logical array is a function \((Array~\mathcal{I}~\mathcal{V})\) that
maps each index \(i \in \mathcal{I}\) to a value
\(v \in \mathcal{V}\).
Operations over arrays are:
\begin{itemize}
\item
  \(select : (Array~\mathcal{I}~\mathcal{V}) \times \mathcal{I}
  \rightarrow \mathcal{V}\) takes an array \(a\) and an index \(i\)
  and returns the value \(v\) stored at index \(i\) in \(a\),
\item
  \(store: (Array~\mathcal{I}~\mathcal{V}) \times \mathcal{I} \times
  \mathcal{V} \rightarrow (Array\ \mathcal{I}\ \mathcal{V})\) takes an
  array \(a\), an index \(i\), and a value \(v\), and returns the
  array \(a\) in which \(i\) maps to \(v\).
\end{itemize}
These functions satisfy the following constraints for all
\({a \in(Array~\mathcal{I}~\mathcal{V})}\), \({i \in \mathcal{I}}\),
\({j \in \mathcal{I}}\), \({v \in \mathcal{V}}\):
\begin{itemize}
\item \(select~(store~a~i~v)~i = v\)
\item \(i \neq j \implies select~(store~a~i~v)~j = select~a~j\)
\end{itemize}
    
\section{Motivating Example}\label{sec:motivating}
Let us consider the toy program in \Cref{list:motivating}. The value
of the conditional at line 3 and the memory access at line 4 are
\emph{leaked}. We say that a \emph{leak is insecure} if it depends on
the secret input. Conversely, a \emph{leak is secure} if it does not
depend on the secret input. CT holds for a program if there is no
insecure leak.

\begin{center}
\begin{minipage}{.9\linewidth}
\begin{pseudocode}[numbers=left,caption={Toy program with one control-flow leak and one memory leak.},label={list:motivating}]
x := private_input();
y := public_input();
if y then return 0;  // leak y = 0
else return tab[x];  // leak x
\end{pseudocode}
\end{minipage}
\end{center}
\vspace{-1em}

Let us take two executions of this program with the same public input:
\((x,y)\) and \((x',y')\) where \(y = y'\). Intuitively, we can see
that the leakages produced at line 3, \(y = 0\) and \(y' = 0\), are
necessarily equal in both executions because \(y = y'\); hence this
leak does not depend on the secret input and is secure. On the
contrary, we can see that the leakages \(x\) and \(x'\) at line 4 can
differ in both executions (e.g.~take \(x := 0\) and \(x' := 1\));
hence this leak depends on the secret input and is insecure.

\emph{The goal of an automatic analysis is to prove that the leak at
line 3 is secure and to return  concrete input showing that the leak
at line 4 is insecure.} 

\myparagraph{Symbolic Execution \& Self-Composition (SC).} %
Symbolic execution can be adapted to the case of CT following the
self-composition principle. Instead of self-composing the program, we
rather self-compose the formula with a renamed version of itself plus
a precondition stating that the low inputs \modif{from each ``part''
  of the formula}{} are equal.
Basically, this amounts to model \emph{two different executions
  following the same path and sharing the same low input} in a single
formula.
At each conditional statement, \emph{exploration queries} are sent to
the solver to determine satisfiable branches -- followed by both
executions (similar to standard SE exploration).
Moreover, additional \emph{insecurity queries} specific to CT are sent
before each conditional statement and memory access to determine if
they depend on the secret or not -- if \modif{one of these extra
  queries}{an insecurity query} is satisfiable then a CT violation is
found.

As an illustration, let us consider the program in
\Cref{list:motivating}. First, we assign symbolic values to
\codeinline{x} and \codeinline{y} and use symbolic execution to
generate a formula of the program until the first conditional (line
3), resulting in the formula:
\(x = \beta ~\wedge~ y = \lambda ~\wedge~ c = (\lambda >
0)\). Second, self-composition is applied on the formula with
precondition \(\lambda = \lambda'\) to constraint the low inputs to be
equal in both executions. Finally, a postcondition \(c \neq c'\) asks
whether the value of the conditional can differ, resulting in the
following insecurity query:
\begin{equation*}
  \lambda = \lambda' ~\wedge~
  \left(\begin{aligned}
      x = \beta ~\wedge~ y = \lambda ~\wedge~ c = (\lambda > 0) ~\wedge~ \\
      x' = \beta' ~\wedge~ y' = \lambda' ~\wedge~ c' = (\lambda' > 0) \\
    \end{aligned}\right)
  ~\wedge~ c \neq c'
\end{equation*}

This formula is sent to a \textsc{smt}-solver. If the solver returns
\textsc{unsat}, meaning that the query is not satisfiable, then the
conditional does not differ in both executions and thus is
secure. Otherwise, it means that the outcome of the conditional
depends on the secret and the solver will return a counterexample
satisfying the insecurity query. %
Here, \texttt{z3} answers that the query is \textsc{unsat} and we can
conclude that the leak is secure. %
With the same method, the analysis finds that the leak at line 4 is
insecure, and returns two inputs (0,0) and (1,0), respectively leaking
0 and 1, as a counterexample showing the violation.

\myparagraph{Limits.} Basic self-composition suffer from two
weaknesses:
\begin{itemize}
\item It generates lots of insecurity queries -- at each conditional
  statement and memory access. Yet, in the previous example it is
  clear that the conditional does not depend on secrets and could be
  spared with better information flow tracking. %
\item The whole original formula is duplicated so the size of the
  self-composed formula is twice the size of the original formula.
  Yet, because the parts of the program that only depend on public
  inputs are equal in both executions, the self-composed formula
  contains redundancies that are not exploited.
\end{itemize}

\myparagraph{Relational Symbolic Execution (RelSE).}  %
We can improve SC by maximizing \emph{sharing} between the pairs of
executions~\cite{palikareva_shadow_2016,farina_relational_2017}.
As previously, RelSE models two executions of a program \(P\) in the
same symbolic execution instance, let us call them \(p\) and \(p'\).
But in RelSE, variables of \(P\) are mapped to \emph{relational
  expressions} %
which are either \emph{pairs} of expressions or \emph{simple}
expressions. The variables that \emph{must be equal} in \(p\) and
\(p'\) (typically, the low inputs) are represented as \emph{simple}
expressions while those that \emph{may be different} are represented
as \emph{pairs} of expressions. First, this enables to share redundant
parts of \(p\) and \(p'\), reducing the size of the self-composed
formula. Second, variables mapping to simple expressions cannot depend
on the secret input, allowing to easily spare some insecurity queries.

As an example, let us perform RelSE of the toy program in
\Cref{list:motivating}. Variable \codeinline{x} is assigned a pair of
expressions ${\pair{\beta}{\beta'}}$ and \codeinline{y} is assigned a
simple expression %
$\simple{\lambda}$. Note that the precondition that public variables
are equal is now implicit since we use the same symbolic variable in
both executions. At line 3, the conditional expression is evaluated to
$c = \simple{\lambda > 0}$ and we need to check that the leakage of
$c$ is secure. Since $c$ maps to a simple expression, we know by
definition that it does not depend on the secret, hence we can spare
the insecurity query.

\emph{RelSE %
  maximizes sharing between both executions and tracks
  secret-dependencies enabling to spare insecurity queries and reduce
  the size of the formula. }

\myparagraph{Challenge of binary-level analysis.} %
Recall that, in binary-level SE, the memory is represented as a
special variable of type \(\memtype\). We cannot directly store
relational expressions in it, so in order to store high inputs at the
beginning of the execution, we have to duplicate it. In other words
the \emph{memory is always duplicated}. %
Consequently, every \(select\) operation will evaluate to a duplicated
expression, preventing to spare queries in many situations.

As an illustration, consider the compiled version of the previous
program, given in~\Cref{lst:limitation_std_relse}. The steps of  
RelSE on this program are given in \cref{fig:relse_motivating2}. %
Note that when the secret input is stored in  memory at line
\((1)\), the array representing the memory is duplicated. This
propagates to the load expression in \texttt{eax} at line \((3)\) and
to the conditional expression at line \((4)\). %
Intuitively, at line \((4)\), \texttt{eax} should be equal to the
simple expression \(\simple{\lambda}\) in which case we could spare
the insecurity query like in the previous example. %
However, because dependencies cannot be tracked in the array
representing the memory, \texttt{eax} evaluates to a pair of
\(select\) expression and we have to send the insecurity query to the
solver.

\begin{center}
\begin{minipage}{.95\linewidth}
  \begin{pseudocode}[numbers=left, caption={Compiled version of the
      conditional in \Cref{list:motivating}, where
      \(\mathtt{x} := \pair{\beta}{\beta'}\) (resp.
      \(\mathtt{x} := \simple{\lambda}\)) denotes that \(\mathtt{x}\)
      is assigned a high (resp. low)
      input.},label={lst:limitation_std_relse}]
@[ebp-8] := $\pair{\beta}{\beta'}$;  // store high input
@[ebp-4] := $\simple{\lambda}$;      // store low input
eax := @[ebp-4];   // assign $\simple{\lambda}$ to eax
ite eax ? $\text{l}_1$ : $\text{l}_2$;     // leak $\simple{\lambda}$
[...]
\end{pseudocode}
\end{minipage}
\end{center}
\vspace{-1em}

\begin{figure}[!htbp]
\centering
\begin{tabular}{@{}rl@{~}r@{~}l@{}}
\((init)\) & \multicolumn{3}{l@{}}{\(\mathtt{mem} := \simple{\mu_0}\)
             and \(\mathtt{ebp} := \simple{ebp}\)}\\
   \((1)\) & \(\mathtt{mem} := \pair{\mu_1}{\mu_1'}\)
             & where & \(\mu_1 \mydef store(\mu_0,ebp-8,\beta)\)\\
           & &   and & \(\mu_1' \mydef store(\mu_0,ebp-8,\beta')\)\\
   \((2)\) & \(\mathtt{mem} := \pair{\mu_2}{\mu_2'}\)
           &  where & \(\mu_2 \mydef store(\mu_1,ebp-4,\lambda)\)\\
           & &  and & \(\mu_2' \mydef store(\mu_1',ebp-4,\lambda)\)\\
   \((3)\) & \(\mathtt{eax} := \pair{\alpha}{\alpha'}\)
           &  where & \(\alpha \mydef select(\mu_2,ebp-4)\)\\
           & &  and & \(\alpha' \mydef select(\mu_2',ebp-4)\)\\
   \((4)\) & \multicolumn{3}{l@{}}{\(leak~\pair{\alpha \neq 0}{\alpha' \neq 0}\)}
\end{tabular}
\caption{RelSE of program in \Cref{lst:limitation_std_relse} where
  \(\mathtt{mem}\) is the memory variable, \(\mathtt{ebp}\) and
  \(\mathtt{eax}\) are registers,
  \(\mu_0, \mu_1, \mu_1', \mu_2, \mu_2'\) are symbolic array
  variables, and \(ebp, \beta, \beta', \lambda, \alpha, \alpha'\) are
  symbolic bitvector variables}\label{fig:relse_motivating2}
\end{figure}

\myparagraph{Practical impact.}  We report in~\cref{tab:motivating} the performances of
CT-analysis on an implementation of  elliptic curve
Curve25519-donna~\cite{bernstein_curve25519:_2006}. %
\textit{Both SC and RelSE fail to  prove the program secure in less than 1h.}
\textit{RelSE} does reduce the number of queries w.r.t.~\textit{SC},
but it is not sufficient.

\begin{table}[!htbp]
\setlength{\tabcolsep}{6pt}
\centering
\begin{tabular}{|l|r|r|r|r|c|}
  \hline
  \multicolumn{1}{|l|}{Version} &
  \multicolumn{1}{c|}{\#I} &
  \multicolumn{1}{c|}{\#Q} &
  \multicolumn{1}{c|}{T} &
  \multicolumn{1}{c|}{\#I/s} &
  \multicolumn{1}{c|}{S}\\
  \hline
  \textit{SC} (e.g. \cite{subramanyan_verifying_2016})
   & 11k & 9051 &   TO & 3 & \hourglass \\
  \textit{RelSE} (e.g. \cite{farina_relational_2017})
  & 13k & 5486 &   TO & 4 & \hourglass \\
  \hline
  \hline
  \brelse{}
  & 10M &  0 & 1166 & 8576 & \ccmark \\
  \hline
\end{tabular}
\caption{Performances of CT-analysis of \texttt{donna} compiled with
  \texttt{gcc-5.4 -O0}, in terms of number of explored unrolled instructions
  (\#I),  number of queries (\#Q), execution time in
  seconds (T), instructions explored per second (\#I/s), and status
  (S) set to secure (\ccmark) or timeout (\hourglass{}) set to 3600s.}
\label{tab:motivating}
\end{table}

\myparagraph{Our solution.} To mitigate this issue, we propose
dedicated simplifications for binary-level relational symbolic
execution that allow a precise tracking of secret-dependencies \emph{in
  the memory} (details in~\cref{sec:optims}). In the particular
example of~\cref{tab:motivating}, our prototype \brelse{} \emph{does
  prove that the code is secure} in less than 20 minutes. Our
simplifications simplify \emph{all the queries}, resulting in a
\(\times 2000\) speedup compared to standard RelSE and SC in terms of
number of instructions treated per second.

\section{Concrete Semantics \& Fault
  Model}\label{sec:concrete_semantics}

\myparagraph{Dynamic Bitvectors Automatas} %
(DBA)~\cite{gopalakrishnan_bincoa_2011} is used by
\binsec{}\cite{david_binsec/se:_2016} as an Intermediate
Representation to model low-level programs and perform its
analysis. The syntax of DBA programs is presented in
\cref{table:dba_lang}.  %
\begin{figure}[htbp]
  \centering
  \fbox{\begin{minipage}{.98\linewidth}
      \[\begin{array}{@{}r@{~}r@{~}l@{~}r@{~}r@{~}l@{}}
    prog  & ::= & \varepsilon ~|~ stmt~prog &
    lval  & ::= &\texttt{v} ~|~ \store~expr \\
    stmt  & ::= &<\texttt{l},inst>  &
    expr  & ::= & \texttt{v} ~|~ \texttt{bv} ~|~ \load~expr \\
    inst  & ::= & lval~\texttt{:=}~expr &
          &   | & \blackdiamond_{u}~expr \\
          &   | & \texttt{ite}~expr\texttt{?}~\texttt{l}_1 \texttt{:} \texttt{l}_2 &
          &   | & expr~\blackdiamond_{b}~expr\\
          &   | & \texttt{goto}~expr ~|~ \texttt{goto l} &
    \blackdiamond_{u} & ::= & \neg ~|~ - \\ 
          &   | & \texttt{halt} &
    \blackdiamond_{b} & ::= & +\: |\: \times\: |\: \leq\: |\: \dots \\
  \end{array}\]\end{minipage}}
  \caption{The syntax of DBA programs, where \texttt{l},
    $\texttt{l}_1$, $\texttt{l}_2$ are  program locations, \texttt{v}
    is a variable and \texttt{bv} is a value.%
  }
  \label{table:dba_lang}
\end{figure}

Let \(\instrset\) denote the set of instructions and \(\locset\) the
set of \modif{}{(program)} locations. %
A program \(\prog{} : \locset \rightarrow \instrset\) is a map from
locations to instructions. %
Values \texttt{bv} and variables \texttt{v} range over the set of
fixed-size bitvectors \(\bvset{n} := {\{0,1\}}^n\) (set of \(n\)-bit
words). %
A concrete configuration is a tuple
\(\cconf{\locvar}{\cregmap}{\cmem}\) where: %
\begin{itemize}
\item \(\locvar \in \locset\) is the current location, and
  \(\locmap{l}\) returns the current instruction,
\item \(\cregmap : \varset{} \to \bvset{n} \) is a register map that maps
  variables to their bitvector value,
\item \(\cmem : \bvset{32} \to \bvset{8}\) is the memory, mapping 32-bit
  addresses to bytes and is accessed by the operator \(@\) (read in an
  expression and write in a left value).
\end{itemize}
The initial configuration is given by
\(\cconfvar_0 \mydef \cconf{\locvar_0}{\cregmap_0}{\cmem_0}\) with
\(\locvar_0\) the address of the entrypoint of the program, \(r_0\) an
arbitrary register map, and \(m_0\) an arbitrary memory.

\myparagraph{Leakage model.}
The behavior of the program is modeled with an instrumented
operational semantics taken from~\cite{barthe_secure_2018} in which
each transition is labeled with an explicit notion of leakage.
A transition from a configuration \(c\) to a configuration \(c'\)
produces a leakage \(\leakvar\), denoted \(c \cleval{\leakvar} c'\). %
Analogously, the evaluation of an expression \(e\) in a configuration
\(\cconf{\locvar}{\cregmap}{\cmem}\), denoted
\(\ceconf{\cregmap}{\cmem}{e} \ceeval{\leakvar} bv\), produces a
leakage \(\leakvar\). %
The leakage of a multistep execution is the concatenation of leakages
produced by individual steps. We use \(\cleval{\leakvar}^k\) with
\(k\) a natural number to denote \(k\) steps in the concrete
semantics.

An excerpt of the concrete semantics is given in
~\cref{fig:dba_semantics} where leakage by memory accesses occur
during execution of load and store instructions and control flow
leakages during execution of dynamic jumps and conditionals. The full
set of rules is given in \cref{app:concrete_evaluation}.

\begin{figure}[htbp]
  \centering

  \fbox{\begin{minipage}{.98\linewidth}\begin{mathpar}

    \inferrule*[left={load}]
    {
      \ceconf{\cregmap}{\cmem}{e} \ceeval{\leakvar} \texttt{bv} \\
    }
    {
      \ceconf{\cregmap}{\cmem}{\load~e} %
      \ceeval{\leakvar \cdot [\texttt{bv}] } \cmem~\texttt{bv} 
    }\and

    \vspace{-1em}

    \inferrule*[lab={d\_jump}]{
      \locmap{l} = \texttt{goto}\ e \\
      \ceconf{\cregmap}{\cmem}{e} \ceeval{\leakvar} \texttt{bv} \\
      \locvar' \mydef to\_loc(\texttt{bv})
    }
    {
      \cconf{\locvar}{\cregmap}{\cmem} %
      \cleval{\leakvar \cdot [\locvar']} %
      \cconf{\locvar'}{\cregmap}{\cmem}
    }\and

    \vspace{-2em}

    \inferrule*[lab={t-ite}]{
      \locmap{l} = \texttt{ite}\ e\ \texttt{?}\ \locvar_1 \texttt{:}\ \locvar_2\\  
      \ceconf{\cregmap}{\cmem}{e} \ceeval{\leakvar} \texttt{bv} \\
      \texttt{bv} \neq 0
    }
    {
      \cconf{\locvar}{\cregmap}{\cmem} %
      \cleval{\leakvar \cdot [\locvar{}_1]} %
      \cconf{\locvar{}_1}{\cregmap}{\cmem}
    }\and

    \inferrule*[left={store}]{
      \locmap{l} = \store{} e := e'\\
      \ceconf{\cregmap}{\cmem}{e} \ceeval{\leakvar} \texttt{bv} \\
      \ceconf{\cregmap}{\cmem}{e'} \ceeval{\leakvar'} \texttt{bv}' \\
    }
    { %
      \cconf{\locvar}{\cregmap}{\cmem} %
      \cleval{\leakvar' \cdot \leakvar \cdot [\texttt{bv}]} %
      \cconf{\locvar+1}{\cregmap}{\cmem[\texttt{bv} \mapsto \texttt{bv'}]} %
    }
  \end{mathpar}\end{minipage}}
\caption{Concrete evaluation of DBA instructions and expressions
  (excerpt), where \(\cdot\) is the concatenation of leakages and
  \(to\_loc : \bvset{32} \to \locset\) converts a bitvector to a
  location.}
  \label{fig:dba_semantics}
\end{figure}

\myparagraph{Secure program.} %
Let \(\highvarset \subseteq \varset\) be the set of high (secret)
variables and \(\lowvarset := \varset \setminus \highvarset\) be the
set of low (public) variables. Analogously, we define
\(\highmemset \subseteq \bvset{32}\) (resp.
\(\lowmemset := \bvset{32} \setminus \highmemset\)) as the addresses
containing high (resp.\ low) input in the initial memory. %

The \emph{low-equivalence relation} over concrete configurations
\(\cconfvar\) and \(\cconfvar'\), denoted
\(\cconfvar \loweq \cconfvar'\), is defined as the equality of low
variables and low parts of the memory. %
Formally, two configurations %
\(\cconfvar \mydef \cconf{\locvar}{\cregmap}{\cmem}\) and %
\(\cconfvar' \mydef \cconf{\locvar'}{\cregmap'}{\cmem'}\) are
low-equivalent iff,
\begin{gather*}
  \forall v \in \lowvarset.\ \cregmap\ v = \cregmap'\ v\ \\
  \forall a \in \lowmemset.\ \cmem\ a = \cmem'\ a
\end{gather*}

\begin{definition}[Constant-time up to \(k\)]
  A program is constant-time (\ct{}) up to \(k\) iff for all
  low-equivalent initial configurations \(\cconfvar_0\) and
  \(\cconfvar'_0\), that
  evaluate in \(k\) steps to \(\cconfvar_k\) and \(\cconfvar'_k\)
  producing leakages \(\leakvar\) and \(\leakvar'\),
  \begin{multline*}
    \cconfvar_0 \loweq \cconfvar_0'\ %
    ~\wedge~ \cconfvar_0 \cleval{\leakvar}^k \cconfvar_k %
    ~\wedge~ \cconfvar'_0 \cleval{\leakvar'}^k \cconfvar'_k%
    \implies \leakvar = \leakvar'
  \end{multline*}
\end{definition}

\section{Binary-level Relational Symbolic
  Execution}\label{sec:std_relse}

Our symbolic execution relies on the
\abv{}~\cite{barrett_smt-lib_2017} first-order
logic. %
We let \(\beta\), \(\beta'\), \(\lambda\), \(\varphi\), \(i\), \(j\)
range over the set of formulas $\formulaset$ in the \abv{} logic.  A
\emph{relational} formula \(\rel{\varphi}\) is either a \abv{} formula
\(\simple{\varphi}\) or a pair \(\pair{\varphi_l}{\varphi_r}\) of two
\abv{} formulas.  We denote \(\lproj{\rel{\varphi}}\)
(resp.\ \(\rproj{\rel{\varphi}}\)), the projection on the left
(resp.\ right) value of \(\rel{\varphi}\). If
\(\rel{\varphi} = \simple{\varphi}\), then \(\lproj{\rel{\varphi}}\)
and \(\rproj{\rel{\varphi}}\) are both defined as \(\varphi\).  We let
\(\rlift{\formulaset}\) be the set of relational formulas and
\(\rlift{\bvtype{n}}\) be the set of relational symbolic bitvectors of
size $n$.

\myparagraph{Symbolic
  configuration.}\label{sec:symbolic-configuration1} %
Since we restrict our analysis to pairs of traces following the same
path -- which is sufficient for constant-time -- the symbolic
configuration only considers a single program location \(l \in Loc\)
at any point of the execution.
\smallskip

\noindent A \emph{symbolic configuration} is of the form
\(\iconfold{l}{\regmap}{\smem}{\pc{}}\) where:
\begin{itemize}
\item \(l \in Loc\) is the current program point,
\item \(\regmap{} : \varset{} \rightarrow \rlift{\formulaset}\) is a
  symbolic register map, mapping variables from a set \(\varset{}\) to
  their symbolic representation as a relational formula in
  \(\rlift{\formulaset}\),
\item \(\smem : \memtype \times \memtype\) is the symbolic memory -- a
  pair of arrays of values in \(\bvtype{8}\) indexed by addresses in
  \(\bvtype{32}\),
\item \(\pc{} \in \formulaset\) is the path predicate -- a conjunction
  of conditional statements and assignments encountered along a path.
\end{itemize}

\myparagraph{Symbolic evaluation} of instructions, denoted
\(\sconfvar \ieval{} \sconfvar'\) where $\sconfvar$ and $\sconfvar'$
are symbolic configurations, is given in
\Cref{fig:eval_instr_sha_excerpt} -- the complete set of rules is
given in \Cref{app:symbolic_evaluation}. %
The evaluation of an expression in a state
\(\stateold{\regmap}{\smem}\) to a relational formula
\(\rel{\varphi}\), is given by
\(\ceconf{\regmap}{\smem}{expr} \eeval{} \rel{\varphi}\). %
A model \(M\) assigns concrete values to symbolic variables. The
satisfiability of a formula \(\pi\) with a model \(M\) is denoted
$M \sat{\pi}$.  Whenever the model is not needed for our purposes, we
leave it implicit and simply write $\sat{\pi}$ for satisfiability.  In
the implementation, we use an \textsc{smt}-solver to determine
satisfiability of a formula.

For the security evaluation of the symbolic leakage we define a
function $\secleak$ which verifies that a relational formula in the
symbolic leakage %
does not differ in its right and left components, i.e.\ that the
symbolic leakage is secure:
\begin{alignat*}{2}
  \secleak(\rel{\varphi} ) &
  &=
  \begin{cases}
    true & {\sf if~} \rel{\varphi} = \simple{\varphi}  \\
    true & {\sf if~} \rel{\varphi} = \pair{\varphi_l}{\varphi_r}
    \wedge \unsat{\big( \pi \wedge (\varphi_l \neq \varphi_r}\big)  \\
    false & {\sf otherwise}
  \end{cases} 
\end{alignat*}

Notice that a simple expression \(\simple{\varphi}\) does not depend
on secrets and can be leaked securely. Thus it \emph{spares an
  insecurity query} to the solver. %
On the other hand, a duplicated expression
\(\pair{\varphi_l}{\varphi_r}\) \emph{may} depend on secrets. Hence
\emph{an insecurity query must be sent to the solver} to ensure that
the leak is secure.

Detailed explanations of the symbolic evaluation rules follow:

\rulename{load} is the evaluation of a load expression. The rule
returns a pair of logical \(select\) formulas from the pair of
symbolic memories \(\smem\) (the box in the hypotheses should be
ignored for now, it will be explained in~\cref{sec:optims}). Note that
the returned expression is \emph{always duplicated} as the \(select\)
must be performed in the left and right memories independently.

\rulename{d\_jump} is the evaluation of a dynamic jump. The rule finds
a concrete value $l'$  for the jump target, and updates the path predicate
and the location. Note that this rule is nondeterministic as \(l'\)
can be any concrete value satisfying the constraint. In practice, we
call the solver to enumerate jump targets up to a given bound and
continue the execution along the valid targets (which jump to an
executable section).

\rulename{ite-true} is the evaluation of a conditional jump when the
expression evaluates to true (the false case is analogous). The rule
updates the path predicate and the next location accordingly.

\rulename{store} is the evaluation of a store instruction. The rule
evaluates the index and value of the store and updates the symbolic
memories and the path predicate with a logical \(store\)
operation. 

\begin{figure}[htbp]
  \fbox{\begin{minipage}{.98\linewidth}\begin{mathpar}
    \vspace{-1em}

    \inferrule*[lab={load}]{
      \econfold{\regmap}{\smem}{e} \eeval{} \rel{\phi}\\
     \boxed{\rel{\varphi} \mydef{} \pair{select(\lproj{\smem}, \lproj{\rel{\phi}}) }{
                                    select(\rproj{\smem}, \rproj{\rel{\phi}})}} \\
                                    \secleak(\rel{\phi})
    }{
      \econfold{\regmap}{\smem}{\load{} e} \eeval{} \rel{\varphi}
    }\and
    
    \vspace{-2em}

    \inferrule*[lab={d\_jump}]{
      \locmap{l} = \texttt{goto}\ e\\
      \econfold{\regmap}{\smem}{e}  \eeval{} \rel{\varphi}  \\
      \pi{}' \mydef{} \pi{} \wedge{} %
      (\lproj{\rel{\varphi{}}} = \rproj{\rel{\varphi{}}})\\
      M\sat{\pi{}'}  \\ l' \mydef M(\lproj{\rel{\varphi{}}}) \\
      \secleak(\rel{\varphi}) }{ \iconfold{l}{\regmap}{\smem}{\pc}
      \ieval{} \iconfold{l'}{\regmap}{\smem}{%
        \pi{}'} }\and

    \vspace{-2em}

    \inferrule*[lab={ite-true}]{
      \locmap{l} = \texttt{ite}\ e\ \texttt{?}\ l_{true} \texttt{:}\ l_{false}\\
      l' \mydef  l_{true} \\
      \econfold{\regmap}{\smem}{e } \eeval{} \rel{\varphi} \\
      \pc{}' \mydef{} \pc{} \wedge{} %
      (true= \lproj{\rel{\varphi{}}} = \rproj{\rel{\varphi{}}})\\
         \sat{\pc'} \\
                \secleak(\rel{\varphi})
   }{
      \iconfold{l}{\regmap}{\smem}{\pc{}} \ieval{}
      \iconfold{l'}{\regmap}{\smem}{\pc{}'}
    }\and    
    
    \vspace{-2em}

    \inferrule*[lab={store}]{
      \locmap{l} = \store{} e := e' \\
      l' = l+1\\
      \econfold{\regmap}{\smem}{e}  \eeval{} \rel{\varphi} \\
      \econfold{\regmap}{\smem}{e'} \eeval{} \rel{\phi} \\
    \smem' \mydef{}  \pair{ store(\lproj{\smem},\lproj{\rel{\varphi}},\lproj{\rel{\phi}})}{store(\rproj{\smem},\rproj{\rel{\varphi}},\rproj{\rel{\phi}})} \\
       \pc'\! \mydef{} \pc{} \wedge \lproj{\smem}'\!=\!store(\lproj{\smem},\lproj{\rel{\varphi}},\lproj{\rel{\phi}})
                          \wedge \rproj{\smem}'\!=\!store(\rproj{\smem},\rproj{\rel{\varphi}},\rproj{\rel{\phi}})\\
                                 \secleak(\rel{\varphi})
    }{
      \iconfold{l}{\regmap}{\smem}{\pc{}} \ieval{}
      \iconfold{l'}{\regmap}{\smem'}{\pc{}'}
    }
  \end{mathpar}\end{minipage}}
  \caption{Symbolic evaluation of DBA instructions and expressions (excerpt). }
  \label{fig:eval_instr_sha_excerpt}
\end{figure}

\myparagraph{Specification of high and low input.} %
By default, the content of the memory and registers is low so we have
to specify addresses that initially contain secret inputs.  The
addresses of high variables can be specified as offsets from the
initial stack pointer \texttt{esp}.  A pair
\(\pair{\beta}{\beta'} \in \rlift{\bvtype{8}}\) of fresh symbolic
variables is stored at each given offset \(h\) and modifies the
symbolic configuration just as a store instruction %
\(\texttt{@[esp + \(h\)] :=}~{\pair{\beta}{\beta'}}\)
would. Similarly, offsets containing low inputs can be set to simple
symbolic expressions \(\simple{\lambda}\) -- although it is not
necessary since the initial memory is equal in both executions.

\myparagraph{Bug-Finding.} A vulnerability is found when the function
\(\secleak(\varphi)\) evaluates to \emph{false}. In this case, the
insecurity query is satisfiable and there exists a model \(M\) such
that %
\(M \sat{\pi \wedge (\lproj{\rel{\varphi}} \neq
  \rproj{\rel{\varphi}})}\).  The model $M$ assigns concrete values to
variables that satisfy the insecurity query. Therefore it can be
returned as a concrete counterexample which triggers the
vulnerability, along with the current location \(\locvar\) of the
vulnerability.

\subsection{Optimizations for binary-level SE}\label{sec:optims}

Relational symbolic execution does not scale in the context of
binary-level analysis (see \textit{RelSE} in
\Cref{tab:scale_total_summary}). In order to achieve better
scalability, we enrich our analysis with an optimization, called
\emph{on-the-fly-read-over-write} (\textit{FlyRow} in
\cref{tab:scale_optims}), based on
\emph{read-over-write}~\cite{farinier_arrays_2018}.  This optimization
simplifies expressions and resolves load operations ahead of the
solver, often avoiding to resort to the duplicated memory and allowing
to spare insecurity queries. %
We also enrich our analysis with two further optimizations, called
\emph{untainting} and \emph{fault-packing} (\textit{Unt} and
\textit{fp} in \cref{tab:scale_optims}), specifically targeting SE for
information flow analysis.

\subsubsection{On-the-Fly Read-Over-Write}\label{sec:row}

Solver calls are the main bottleneck of symbolic execution, and
reasoning about \(store\) and \(select\) operations in arrays is
particularly challenging~\cite{farinier_arrays_2018}. Read-over-write
(Row)~\cite{farinier_arrays_2018} is a simplification for the theory
of arrays that efficiently resolves \(select\) operations. This
simplification is particularly efficient in the context of
binary-level analysis because the memory is represented as an array
and formulas contain many \(store\) and \(select\) operations.

The standard read-over-write optimization~\cite{farinier_arrays_2018}
has been implemented as a solver-pre-processing, simplifying a formula
before sending it to the solver. While it has proven to be very
efficient to simplify individual formulas of a single
execution~\cite{farinier_arrays_2018}, we show in \cref{sec:comp-se}
that it does not scale in the context of relational reasoning, where
formulas model two executions and a lot of queries are sent to the
solver. %

Thereby, we introduce \emph{on-the-fly read-over-write}
(\textit{FlyRow}) to track secret-dependencies in the memory and spare
insecurity queries in the context of information flow analysis. By
keeping track of \emph{relational \(store\)} expressions along the
symbolic execution, it can resolve \(select\) operations -- often
avoiding to resort to the duplicated memory -- and drastically reduces
the number of queries sent to the
solver, %
improving the performances of the analysis.

\myparagraph{Lookup.} %
The symbolic memory can be seen as the history of the successive
\(store\) operations beginning with the initial memory \(\mu_0\). %
Therefore, a memory \(select\) can be resolved by going back up the
history and comparing the index to load, with indexes previously
stored. %
Our optimization consists in replacing selection in the memory
(\Cref{fig:eval_instr_sha_excerpt}, \rulename{load} rule, boxed
hypothesis) by a new function %
\(\lookup : (\memtype \times \memtype) \times \bvtype{32} \to
\rlift{\bvtype{8}}\) which takes a relational memory and an index, and
returns the relational value stored at that index. The lookup function
can be lifted to relational indexes but for simplicity we only define
it for simple indexes and assume that relational store operations
happen to the same index in both sides -- note that for constant-time
analysis, this hypothesis holds. The function returns a relational
bitvector formula, and is defined as follows:
\begin{alignat*}{2}
  &\lookup(\smem_0, i) = ~~\pair{select(\lproj{{\smem_0}}, i)}{select(\rproj{{\smem_0}}, i)}\\
  &\lookup(\smem_n, i) = \\
 &
  \begin{cases}
    \simple{\varphi_l}  & {\sf if~} \compare(i,j)  \wedge \compare(\varphi_l,\varphi_r) \\
    \pair{\varphi_l}{\varphi_r}  & {\sf if~} \compare(i,j) \wedge \neg\compare(\varphi_l,\varphi_r)\\
    \lookup(\smem_{n-1}, i)& {\sf if~} \neg\compare(i,j)\\ 
    \rel{\phi}& {\sf if~} \compare(i,j) = \bot
  \end{cases} \\  \text{~where~}\\
   & \smem_n \mydef \pair{store(\lproj{{\smem_{n-1}}},j,\varphi_l)}{store(\rproj{{\smem_{n-1}}},j,\varphi_r)} \\
   & \rel{\phi} \mydef \pair{select(\lproj{{\smem_n}}, i)}{select(\rproj{{\smem_n}}, i)}
  \end{alignat*}
  where \(\compare(i,j)\) is a comparison function relying on
  \emph{syntactic term equality}, which returns true (resp.\ false)
  only if \(i\) and \(j\) are equal (resp.\ different) in any
  interpretation. If the terms are not comparable, it is undefined,
  denoted \(\bot\).

\begin{example}[Lookup]\label{ex:lookup}
  Let us consider the memory: \tikzstyle{box1}=[draw, fill=gray!10, text centered, minimum width = 2em, minimum height=1.8em,scale=0.8]
\tikzstyle{box2}=[draw, text centered, minimum width = 1.5em, minimum height=1.8em,scale=0.8]
\begin{center}
  \begin{tikzpicture}[node distance=0cm,outer sep = 0pt]
    \node (I) [] {};
    \node (L) [left=0pt of I] {\(\smem ~~=\)};
    \node (S)  [right=0pt of L] {};
    \node (C)  [box1,right = 0pt of S] {\(ebp-4\)};
    \node (C1) [box2,anchor=north west] at (C.north east) {$\simple{\lambda}$};
    \node (B)  [box1,right = 1.8em of C1] {\(ebp-8\)};
    \node (B1) [box2,anchor=north west] at (B.north east) {$\pair{\beta}{\beta'}$};
    \node (D)  [box1,right = 1.8em of B1] {\(esp\)};
    \node (D1) [box2,anchor=north west] at (D.north east) {$\simple{ebp}$};
    \node (A)  [right = 1.8em of D1] {\([~]\)};
    \draw[->] (C1) -- (B);
    \draw[->] (B1) -- (D);
    \draw[->] (D1) -- (A);
  \end{tikzpicture}
\end{center}

  \begin{itemize}
  \item A call to \(\lookup(\rel{\mu}, ebp - 4)\) returns \(\lambda\).
  \item A call to \(\lookup(\rel{\mu}, ebp - 8)\) \modif{starts by
      comparing}{first compares} the indexes \([ebp-4]\) and
    \([ebp-8]\). Because it can determine that these indexes are
    \emph{syntactically distinct}, the function moves to the second
    element, determines the syntactic equality of indexes and returns
    \(\pair{\beta}{\beta'}\).
  \item A call to \(\lookup(\rel{\mu}, esp)\) tries to compare the
    indexes \([ebp-4]\) and \([esp]\). Without further information,
    \modif{it cannot compare \(ebp\) and \(esp\), thus the equality or
      disequality of these indexes cannot be determined. In this
      case,}{the equality or disequality of \(ebp\) and \(esp\)
      cannot be determined, therefore} the lookup is aborted and the
    \(select\) operation cannot be simplified.
  \end{itemize}
\end{example}

\myparagraph{Term rewriting.} %
To improve the conclusiveness of this syntactic comparison, the terms
are assumed to be in \emph{normalized} form \(\beta + o\) where
\(\beta\) is a base (i.e.\ an expression on symbolic variables) and
\(o\) is a constant offset.  In order to apply \textit{FlyRow}, we
normalize all the formulas created during the symbolic execution
(details of our normalization function are ommited for space reasons).
The comparison of two terms \(\beta + o\) and \(\beta' + o'\) in
normalized form can be efficiently computed as follows: if the bases
\(\beta\) and \(\beta'\) are syntactically equal, then return
\(o = o'\), otherwise the terms are not comparable.

In order to increase the conclusiveness of \textit{FlyRow}, we also
need variable inlining. However, inlining all variables is not a
viable option as it would lead to an exponential term size growth. %
Instead, we define a \emph{canonical form} \(v + o\) where \(v\) is a
bitvector variable, and \(o\) is a constant bitvector offset, and we
only inline formulas that are in canonical form. It enables rewriting
of most of the memory accesses on the stack which are of the form
\codeinline{ebp + bv} while avoiding term-size explosion.

\subsubsection{Untainting}
\label{sec:untainting}
After the evaluation of a rule with the predicate $\secleak$ for a
duplicated expression \(\pair{\varphi_l}{\varphi_r} \), we know that
the equality \(\varphi_l = \varphi_r\) holds in the current
configuration. From this equality, we can deduce useful information
about variables that must be equal in both executions. We can then
propagate this information to the register map and memory in order to
spare subsequent insecurity queries concerning these variables. %
For instance, consider the leak of the duplicated expression
\(\pair{v_l + 1}{v_r + 1}\), where \(v_l\) and \(v_r\) are symbolic
variables. If the leak is secure, we can deduce that \(v_l = v_r\) and
replace all occurrences of \(v_r\) by \(v_l\) in the rest of the
symbolic execution.

We define a function \(\untaint(\regmap,\smem, \rel{\varphi})\) that
takes a register map \(\regmap\), a memory \(\smem\), and a duplicated
expression \(\rel{\varphi}\); it applies the rules defined
in~\cref{fig:untainting_rules} \modif{to}{which} deduce variable
equalities from \(\rel{\varphi}\), propagate them in \(\regmap\) and
\(\smem\), and return a pair of updated register map and memory
\((\regmap', \smem')\). %
Intuitively, if the equality of variables \(v_l\) and \(v_r\) can be
deduced from \(\secleak(\rel{\varphi})\), the \(untaint\) function
replaces occurences of \(v_r\) by \(v_l\) in the memory and the
register map. As a result, a duplicated expression \(\pair{v_l}{v_r}\)
would be replaced by the simple expression \(\simple{v_l}\) in the
rest of the execution\footnote{We implement untainting with a cache of
  "untainted variables" that are substituted in the program copy when
  relational expressions are built.}.%

\begin{figure}
  \centering
  \begin{gather*}
    untaint(\regmap,\smem,\pair{v_l}{v_r}) = (\regmap[v_r \backslash v_l],\smem[v_r \backslash v_l])\\
    \begin{rcases}
    untaint(\regmap,\smem,\pair{\neg t_l}{\neg t_r})\\
    untaint(\regmap,\smem,\pair{-t_l}{-t_r})\\
    untaint(\regmap,\smem,\pair{t_l + k}{t_r + k})\\
    untaint(\regmap,\smem,\pair{t_l - k}{t_r - k})\\
    untaint(\regmap,\smem,\pair{t_l :: k}{t_r :: k})\\
  \end{rcases} = untaint(\regmap,\smem,\pair{t_l}{t_r})
  \end{gather*}
  \caption{Untainting    rules where \(v_l, v_r\) are bitvector variables
   and \(t_l, t_r, k\) are arbitrary bitvector terms, and
   $f[v_r \backslash v_l]$ 
   indicates that the variable \(v_r\) is substituted with 
   \(v_l\) in   $f$. }  %
  \label{fig:untainting_rules}
\end{figure}

\subsubsection{Fault-Packing}
\label{sec:fp}
\modif{The number of insecurity checks along symbolic execution is
  important for CT.}{For CT, the number of insecurity checks generated
  along the symbolic execution is substantial.}  The fault-packing
(\textit{fp}) optimization gathers these insecurity checks along a
path and postpones their resolution to the end of the basic block.

\begin{example}[Fault-packing]
  For example, let us consider a basic-block with a path predicated
  \(\pc\). If there are two memory accesses along the basic block that
  evaluate to \(\pair{\varphi}{\varphi'}\) and \(\pair{\phi}{\phi'}\),
  we would normally generate two insecurity queries
  \((\pc \wedge \varphi \neq \varphi')\) and
  \((\pc \wedge \phi \neq \phi')\) -- one for each memory access.
  \textit{fp} regroups these checks into a single query
  \(\big(\pc \wedge ((\varphi \neq \varphi') \lor (\phi \neq
  \phi'))\big)\) sent to the solver at the end of the basic block.
\end{example}

This optimization reduces the number of insecurity queries sent to the
solver and thus helps improving performance. However it degrades the
precision of the counterexample: while checking each instruction
individually precisely points to vulnerable instructions,
fault-packing reduces accuracy to vulnerable basic blocks only.  Note
that even though disjunctive constraints are usually harder to solve
than pure conjunctive constraints, those introduced by \textit{fp} are
very limited (no nesting) and thus do not add much
complexity. Accordingly, they never end up in a performance
degradation (see~\cref{tab:scale_optims}).

\subsection{Theorems}~\label{sec:proofs} In order to define properties
of our symbolic execution, we use $\cleval{}^k$ (resp.\ $\ieval{}^k$),
with $k$ a natural number, to denote $k$ steps in the concrete
(resp.\ symbolic) evaluation.

\begin{definition}[\(\concsym{p}{M}\)]\label{def:concsym}
  We define a concretization relation $\concsym{p}{M}$ between
  concrete and symbolic configurations, where \(M\) is a model and
  \(p \in \{l,r\}\) is a projection on the left or right side of a
  symbolic configuration.  Intuitively, the relation
  $c\! \concsym{p}{M}\! s$ is the concretization of the \(p\)-side of
  the symbolic state \(s\) with the model \(M\).
  Let \(c \mydef \cconf{\locvar_1}{\cregmap}{\cmem}\) and
  \(s \mydef \iconfold{\locvar_2}{\regmap}{\smem}{\pc}\). Formally
  $c\! \concsym{p}{M}\! s$ holds iff \(M\!\sat\!\pc\),
  \(\locvar_1 = \locvar_2\) and for all expression \(e\), either the
  symbolic evaluation of \(e\) gets stuck or we have
  \[\econfold{\regmap}{\smem}{e} \eeval{} \rel{\varphi} ~\wedge~ %
  (M(\proj{\rel{\varphi}}) = \mathtt{bv} \iff c~e \ceeval{}
  \mathtt{bv})\] %

  Notice that because both executions represented in the initial
  configuration \(s_0\) are low-equivalent,
  \(\cconfvar_0 \concsym{l}{M} s_0 ~\wedge~ \cconfvar'_0
  \concsym{r}{M} s_0\) implies that
  \(\cconfvar_0 \loweq \cconfvar_0'\).
\end{definition}

Through this section, we assume that the program \(\prog{}\) is
defined on all locations computed during the symbolic execution. Under
this hypothesis, the symbolic execution can only get stuck when an
expression \(\rel{\varphi}\) is leaked such that
\(\neg \secleak(\rel{\varphi})\). In this case, a vulnerability is
detected and there exists a model $M$ such that
\({M \sat \pc \wedge (\lproj{\rel{\varphi}} \neq
  \rproj{\rel{\varphi}})}\). %

\smallskip The following theorem claims the completeness of our
symbolic execution relatively to an initial symbolic state. If the
program is constant-time up to \(k\), then for each pair of concrete
executions up to \(k\), there exists a corresponding symbolic
execution (no under-approximation). A \emph{proof} is given
in \cref{app:completeness}.
\begin{restatable}[Relative Completeness of
  RelSE]{theorem}{completeness}\label{thm:completeness}
  Let \(P\) be a program constant-time up to \(k\) and $s_0$ be a
  symbolic initial configuration for $P$. For every concrete states
  $c_0$, $c_k$, $c_0'$, $c_k'$, and model $M$ such that
  ${c_0 \concsym{l}{M} s_0} ~\wedge~ {c_0' \concsym{r}{M} s_0}$, %
  if $c_0 \cleval{\leakvar}^k c_k$ and
  $c_0' \cleval{\leakvar'}^k c_k'$ with \(\leakvar = \leakvar'\) then
  there exists a symbolic configuration \(s_k\) and a model \(M'\)
  such that: %
  \[s_0 \ieval{}^k s_k ~\wedge~ %
    c_k \concsym{l}{M'} s_k ~\wedge~ c_k' \concsym{r}{M'} s_k\]
\end{restatable}

The following theorem claims the correctness of our symbolic
execution, stating that for each symbolic execution and model \(M\)
satisfying the path predicate, the concretization of the symbolic
execution with \(M\) corresponds to a valid concrete execution (no
over-approximation).  A \emph{proof} is given in
\cref{app:correctness}.

\begin{restatable}[Correctness of RelSE]{theorem}{correctness}\label{thm:correctness}
  For every symbolic configurations $s_0$, $s_k$ such that
  \(s_0 \ieval{}^k s_k\) and for every concrete configurations
  \(c_0\), \(c_k\) and model \(M\), such that
  \(c_0 \concsym{p}{M} s_0\) and \(c_k \concsym{p}{M} s_k\),
  there exists a concrete execution \(c_0 \cleval{}^k c_k\).
\end{restatable}

The following is our main result. If the symbolic execution does not
get stuck due to a satisfiable insecurity query, then the program is
constant-time. The \emph{proof} is given in \cref{app:ct-security}.
\begin{restatable}[Bounded-Verification for
  CT]{theorem}{security}\label{thm:bv}
  Let $s_0$ be a symbolic initial configuration for a program $P$. If
  the symbolic evaluation does not get stuck, then $P$ is
  constant-time w.r.t.\ $s_0$. Formally, if for all $k$,
  $s_0 \ieval{}^k s_k$ then for all initial configurations
  \(\cconfvar_0\) and \(\cconfvar_0'\) and model \(M\) such that
  \(\cconfvar_0 \concsym{l}{M} s_0\), and
  \(\cconfvar'_0 \concsym{r}{M} s_0\),
  \begin{multline*}
    \cconfvar_0 \loweq \cconfvar_0'\ ~\wedge~ %
    \cconfvar_0  \cleval{\leakvar}^k \cconfvar_k ~\wedge~ %
    \cconfvar_0' \cleval{\leakvar'}^k \cconfvar'_k %
    \implies \leakvar = \leakvar'
  \end{multline*}
  Additionally, if \(s_0\) is fully symbolic, then \(P\) is
  constant-time.
\end{restatable}

The following theorem expresses that when the symbolic execution gets
stuck, then there is a concrete path that violates constant-time. 
The \emph{proof} is given in \cref{app:ct-bf}. 
\begin{restatable}[Bug-Finding for CT]{theorem}{bugfinding}\label{thm:bf}
  Let $s_0$ be an initial symbolic configuration for a program $P$. If
  the symbolic evaluation gets stuck in a configuration \(s_k \) then
  $P$ is not constant-time. Formally, if there exists $k$ st.
  \(s_0 \ieval{}^k s_k\) and \(s_k\) is stuck, then there exists a
  model \(M\) and concrete configurations
  \(\cconfvar_0 \concsym{l}{M} s_0\), %
  \(\cconfvar_0' \concsym{r}{M} s_0 \), %
  \(\cconfvar_k \concsym{l}{M} s_k \) and %
  \(\cconfvar_k' \concsym{r}{M} s_k\) such that, %
  \begin{multline*}%
    \cconfvar_0 \loweq \cconfvar_0' ~\wedge~%
    \cconfvar_0  \cleval{\leakvar}^k \cconfvar_k ~\wedge~ %
    \cconfvar_0' \cleval{\leakvar'}^k \cconfvar'_k %
    \wedge \leakvar \neq \leakvar' %
  \end{multline*}
\end{restatable} 

\section{Implementation}
\label{sec:implem}

We implemented our relational symbolic execution, \brelse{}, on top of
the binary-level analyzer \binsec{}~\cite{david_binsec/se:_2016}.
\brelse{} takes as input a x86 or ARM executable, a specification of
high inputs and an initial memory configuration (possibly fully
symbolic). \modif{It then}{It} performs bounded exploration of the
program under analysis (up to a user-given depth), and reports the
identified CT violations together with counterexemples (i.e., initial
configurations leading to the vulnerabilities).
In case no violation is reported, if the initial configuration is
fully symbolic and the program has been explored exhaustively then the
program is \emph{proven} secure.

\brelse{} is composed of a \emph{relational symbolic exploration}
module and an \emph{insecurity analysis} module. The symbolic
exploration module chooses the path to explore, updates the symbolic
configuration, builds the path predicate and ensure that it is
satisfiable. The insecurity analysis module builds insecurity queries
and check that they are not satisfiable. 

We explore the program in a depth-first search manner and we rely on
the Boolector \textsc{smt}-solver ~\cite{niemetz_boolector_2014},
currently the best on theory
\abv{}~\cite{noauthor_smt-comp_nodate,farinier_arrays_2018}.

\myparagraph{Overall architecture} is illustrated in
\cref{fig:archi}. The \textsc{Disasm} module loads the executable and
lifts the machine code to the DBA intermediate representation
\cite{gopalakrishnan_bincoa_2011}. Then, the analysis is performed
\modif{}{by the \textsc{Rel} module} on the DBA code.
\modif{The \textsc{Rel} plugin interacts with the \textsc{smt}-solver
  through the interface}{The} \textsc{Formula} module
\modif{which}{} is in charge of \modif{building the queries, and pass
  them to}{building and simplifying formulas, and sending the queries
  to} the \textsc{smt}-solver.  The queries are exported to the
SMTLib~\cite{barrett_smt-lib_2017} standard which permits to interface
with many off-the-shelf \textsc{smt}-solvers.
The \modif{plugin}{\textsc{Rel} plugin} represents \(\approx\)3.5k
lines of Ocaml.

\begin{figure}[!tbp]
  \centering
  \includegraphics[width=.9\linewidth]{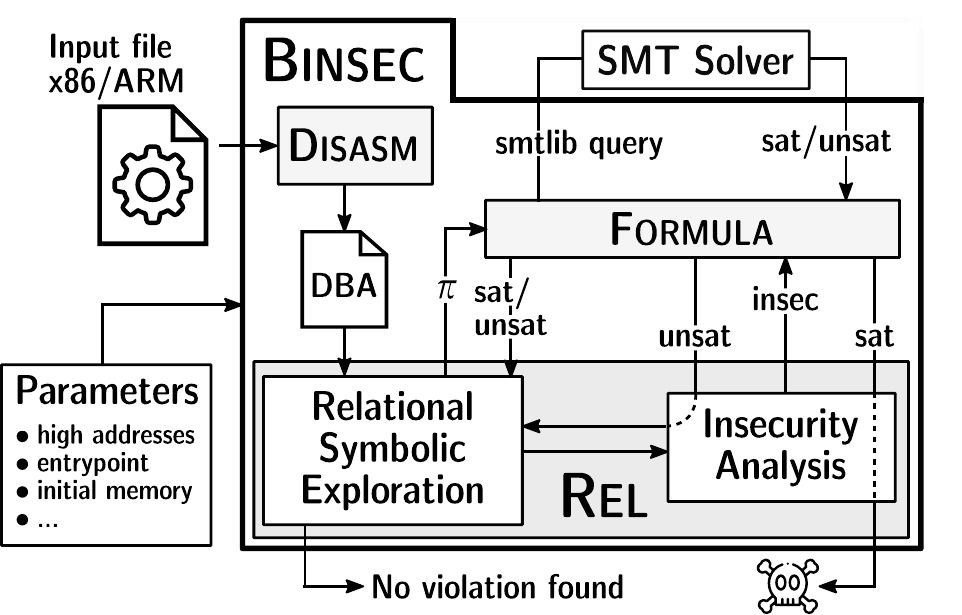}
  \caption{\binsec{} architecture with \brelse{} plugin.}
  \label{fig:archi}
\end{figure}

\myparagraph{Usability.}
\label{sec:usability}
Binary-level semantic analyzers tend to be harder to use than their
source-level counterparts as inputs are more difficult to specify and
results more difficult to interpret. In order to mitigate this point,
we propose a vizualisation mechanism (based on IDA, which highlight
coverage and violations) and easy input specification (using dummy
functions, cf.~\cref{app:stubs}) when source-level information is
available.

\section{Experimental Results}
\label{sec:expes}

We answer the following research questions:
\begin{description}
\item[RQ1: Effectiveness] Is \brelse{} able to perform constant-time
  analysis on real cryptographic \modif{binary codes}{binaries}, for
  both bug finding and bounded-verification?

\item[RQ2: Genericity] Is \brelse{}  generic enough to encompass several
  architectures and compilers?

\item[RQ3: Comparison vs. Standard Approaches] How does
  \brelse{} scale compared to traditional approaches based on standard SC and RelSE? 

\item[RQ4: Impact of simplifications] What are the respective  impacts of 
  our different simplifications?

\item[RQ5: Comparison vs. SE] What is the overhead of \brelse{} compared
  to  standard SE, and can our 
  simplifications be useful for standard SE? 
\end{description}

Experiments were performed on a laptop with an Intel(R) Core(TM)
i5-2520M CPU @ 2.50GHz processor and 32GB of RAM, running Linux Mint
18.3 Sylvia. %
Similarly to related work (e.g.~\cite{doychev_cacheaudit:_2015}),
\texttt{esp} is initialized to a concrete value, we start the analysis
from the beginning of the \texttt{main} function, we statically
allocate data structures and the length of keys and buffers is fixed
(e.g.~for Curve25519-donna~\cite{bernstein_curve25519:_2006}, three
256-bit buffers are used to store the input, the output and the secret
key). When not stated otherwise, programs are compiled for x86 (32bit)
with their default compiler setup.

\subsection{Effectiveness  (RQ1,RQ2)}
\label{sec:effectiveness}
We carry out three experiments to assess the effectiveness of our
technique: (1) bounded-verification of secure cryptographic primitives
previously verified at source- or LLVM-
level~\cite{blazy_verifying_2017,almeida_verifying_2016,zinzindohoue_hacl*:_2017},
(2) automatic replay of known bug
studies~\cite{simon_what_2018,almeida_verifying_2016,al_fardan_lucky_2013},
(3) automatic study of CT preservation by compilers extending prior
work~\cite{simon_what_2018}. %
Overall, our study encompasses 338 representative code samples for a total of 70k machine instructions  and 22M  unrolled instructions  
(i.e., instructions explored by \brelse{}).   

\myparagraph{Bounded-Verification (RQ1).} 
\label{sec:bounded-verification}
We analyze a large range of \emph{secure} constant-time cryptographic
primitives (296 samples, 64k instructions), comprising:
(1) several basic constant-time utility functions such as selection
functions~\cite{simon_what_2018}, sort
functions~\cite{noauthor_imdea-software/verifying-constant-time_nodate}
and utility functions from
HACL*\footnote{\url{https://github.com/project-everest/hacl-star/blob/master/snapshots/hacl-c/Hacl_Policies.c}
  and
  \url{https://github.com/project-everest/hacl-star/blob/master/snapshots/hacl-c/kremlib_base.h}\label{fnote:hacl}}
and
OpenSSL\footnote{\url{https://github.com/xbmc/openssl/blob/master/crypto/constant_time_locl.h}\label{fnote:openssl}};
(2) a set of representative constant-time cryptographic primitives
already studied in the literature on source
code~\cite{blazy_verifying_2017} or
LLVM~\cite{almeida_verifying_2016}, including implementations of
TEA~\cite{goos_tea_1995},
Curve25519-donna%
~\cite{bernstein_curve25519:_2006}, \texttt{aes} and \texttt{des}
encryption functions taken from BearSSL~\cite{pornin_bearssl_nodate},
cryptographic primitives from libsodium~\cite{hutchison_security_2012}
and the constant-time padding remove function
\texttt{tls-cbc-remove-padding} from
OpenSSL~\cite{almeida_verifying_2016};
(3) a set of functions from the HACL*
library~\cite{zinzindohoue_hacl*:_2017}.

Results are reported in~\cref{tab:bounded-verif}. For each program,
\brelse{} is able to perform an exhaustive exploration without finding
any violations of constant-time in less than 20 minutes. Note that
exhaustive exploration is possible because in cryptographic programs,
bounding the input size bounds loops. These results show that
\brelse{} can perform bounded-verification of real-world cryptographic
implementations at binary-level in a reasonable time, which was
impractical with previous approaches based on self-composition or
standard RelSE (see~\cref{sec:scalability}).

\emph{This is the first automatic CT-analysis of these cryptographic
  libraries at the binary-level.}

\begin{table}[!htbp]
  \centering
  \begin{tabular}{|lr|r|r|r|c|}
    \cline{3-6}
    \multicolumn{2}{l|}{}
    & \multicolumn{1}{c|}{\(\approx\text{\#I}\)}
    & \multicolumn{1}{c|}{\#\(\text{I}_{u}\)}
    & \multicolumn{1}{c|}{T}
    & \multicolumn{1}{c|}{S}
    \\ \hline
    \multirow{4}{*}{utility}
    & ct-select   & 1015 &  1507 &  .21 & 29 \(\times\) \ccmark{} \\
    & ct-sort & 2400 &  1782 &  .24 & 12 \(\times\) \ccmark{} \\
    & Hacl*   & 3850 & 90953 & 9.34 & 110 \(\times\) \ccmark{} \\
    & OpenSSL & 4550 &  5113 &  .75 & 130 \(\times\) \ccmark{} \\   
    \hline
    \multirow{2}{*}{tea}
        & \texttt{-O0} & 290 & 953 & .12 & \ccmark{} \\
        & \texttt{-O3} & 250 & 804 & .12 & \ccmark{} \\
    \hline
    \multirow{2}{*}{donna}
        & \texttt{-O0} & 7083 & 10.2M & 1166 & \ccmark{} \\
        & \texttt{-O3} & 4643 &  2.7M &  401 & \ccmark{} \\
    \hline
    \multirow{4}{*}{libsodium}
    & salsa20  &  1627 &  6.5k &  .7 & \ccmark{} \\
    & chacha20 &  2717 & 30.0k & 5.0 & \ccmark{} \\
    & sha256   &  4879 & 38.7k & 4.5 & \ccmark{} \\
    & sha512   & 16312 & 62.1k & 7.1 & \ccmark{} \\
    \hline
    \multirow{4}{*}{Hacl*}
    & chacha20   & 1221 &  5.0k &  1.0 & \ccmark{} \\
    & curve25519 & 8522 &  9.4M & 1110 & \ccmark{} \\
    & sha256     & 1279 & 16.8k &  2.8 & \ccmark{} \\
    & sha512     & 2013 & 31.8k &  4.3 & \ccmark{} \\   
    \hline
    \multirow{2}{*}{BearSSL}
    & aes\_ct   & 357 &  3.5k  &   .6 & \ccmark{} \\
    & des\_ct   & 682 & 38.5k  & 33.9 & \ccmark{} \\
    \hline
    \multirow{1}{*}{OpenSSL}
    & tls-rempad-patch & 424 & 35.7k & 406 & \ccmark{} \\
    \hline
    \hline
    \textbf{Total} &  & 64114 & 22.7M &  3154 & 296 \(\times\) \ccmark{} \\
    \hline
  \end{tabular}
  \caption{Bounded verification of constant-time cryptographic
    implementations where \#\(\text{I}\) (resp. \#\(\text{I}_{u}\))
    is the number of static (resp. unrolled) instructions, T is the
    execution time in seconds, and S is the status (\hourglass{} for
    timeout or \ccmark{} for exhaustive
    exploration).}\label{tab:bounded-verif}
\end{table}

\myparagraph{Bug-Finding (RQ1).}
\label{sec:bug-finding}
We take three known bug studies from the literature
\cite{simon_what_2018,noauthor_imdea-software/verifying-constant-time_nodate,al_fardan_lucky_2013}
and replay them automatically at binary-level (42 samples, 6k
instructions), including:
(1) binaries compiled from constant-time sources of a selection
function~\cite{simon_what_2018} and sort
functions~\cite{noauthor_imdea-software/verifying-constant-time_nodate},
(2) non-constant-time versions of \texttt{aes} and \texttt{des} from
BearSSL~\cite{pornin_bearssl_nodate},
(3) the non-constant-time version of OpenSSL's
\texttt{tls-cbc-remove-padding}\footnote{\url{https://github.com/openssl/openssl/blob/OpenSSL_1_0_1/ssl/d1_enc.c}\label{fnote:lucky13}}
responsible for the famous Lucky13 attack~\cite{al_fardan_lucky_2013}.

Results are reported in~\cref{tab:bug-finding} with \emph{fault-packing
  disabled} to report vulnerabilities at the instruction level. All
bugs have been found within the timeout.
Interestingly, we found 3 \emph{unexpected binary-level vulnerabilities
  (from secure source codes) that slipped through previous analysis}:

\begin{itemize}
\item function \texttt{ct\_select\_v1}  \cite{simon_what_2018} was deemed secured through binary-level manual inspection, 
still  we confirm  
that any version of \texttt{clang} with \texttt{-O3} introduces a
secret-dependent conditional jump which violates constant-time; 

\item functions \texttt{ct\_sort} and \texttt{ct\_sort\_mult},
  verified by ct-verif~\cite{almeida_verifying_2016} (LLVM bytecode
  compiled with \texttt{clang}), are vulnerable when compiled with
  \texttt{gcc -O0} or \texttt{clang -O3 -m32 -march=i386}.
\end{itemize}

\noindent A few more details on these vulnerabilities are provided in
the next study.  Finally, we describe the application of \brelse{} to
the Lucky13 attack in \cref{app:lucky13}.

\bgroup
\footnotesize
\setlength{\tabcolsep}{1.5pt}
\begin{table}[!htbp]
  \centering
  \begin{tabularx}{\linewidth}{|lr|r|r|r|c|c|r|C|}
    \cline{3-9}
    \multicolumn{2}{l|}{}
    & \multicolumn{1}{c|}{\(\approx\text{\#I}\)}
    & \multicolumn{1}{c|}{\#\(\text{I}_{u}\)}
    & \multicolumn{1}{c|}{T}
    & \begin{tabular}{@{}c@{}}CT\\src\end{tabular}
    & \begin{tabular}{@{}c@{}}S\end{tabular}
    & \multicolumn{1}{c|}{\bug{}}
    & Comment
    \\ \hline
    \multirow{2}{*}{utility}
    & ct-select &  735 &  767 &  .29 & Y & 21\(\times\)\cxmark & 21 & 1 new \cxmark{}\\
    & ct-sort   & 3600 & 7513 & 13.3 & Y & 18\(\times\)\cxmark & 44 & 2 new \cxmark{}\\
    \hline
    \multirow{2}{*}{BearSSL}
    & aes\_big & 375 &    873 & 1574 & N & \cxmark & 32 & - \\
    & des\_tab & 365 &  10421 &  9.4 & N & \cxmark &  8 & - \\
    \hline
    \multicolumn{2}{|@{}l@{}|}{\multirow{1}{*}{
    \begin{tabular}{@{}lr@{}}
      \multirow{1}{*}{OpenSSL}
      & tls-rempad-luk13 \\
    \end{tabular}
    }}
       & 950 &  11372 & 2574 & N & \cxmark & 5 & - \\
    \hline
    \hline
    \textbf{Total} &  & 6025 & 30946 & 4172 & - & 42\(\times\)\cxmark & 110 & - \\
    \hline
  \end{tabularx}
  \caption{Bug-finding of constant-time in cryptographic
    implementations where \#\(\text{I}\) (resp. \#\(\text{I}_{u}\))
    is the number of static (resp. unrolled) instructions, T is the
    execution time in seconds, CT src means that the source 
    is constant-time, S is the status (\cxmark{}
    for insecure program), and \bug{} is the number of
    bugs.}\label{tab:bug-finding}
\end{table}
\egroup

\myparagraph{Effects of compiler optimizations on CT (RQ1, RQ2).}  
\label{sec:compilers}
Simon \emph{et al.}~\cite{simon_what_2018} \emph{manually} analyse
whether \texttt{clang} optimizations break the constant-time property,
for 5 different versions of a selection function. %
We reproduce their analysis in an \emph{automatic} manner and {\it
  extend it significantly}, adding:
29 new functions, a newer version of \texttt{clang}, the ARM
architecture, the \texttt{gcc} compiler and
\texttt{arm-linux-gnueabi-gcc} version 5.4.0 for ARM -- for a total of
408 executables (192 in the initial study).
Results are
presented in \cref{tab:verif_compilers}. 

\begin{table}[!htbp]
 \setlength{\tabcolsep}{1.6pt}
 \centering
  \begin{tabularx}{\linewidth}{|X|cc|cc|cc||cc|cc||cc|}
    \cline{2-13}
    \multicolumn{1}{c|}{}
    & \multicolumn{2}{c|}{cl-3.0}
    & \multicolumn{2}{c|}{cl-3.9}
    & \multicolumn{2}{c||}{\textbf{cl-7.1}}
    & \multicolumn{2}{c|}{\textbf{gcc-5.4}}
    & \multicolumn{2}{c||}{\textbf{gcc-8.3}}
    & \multicolumn{2}{c|}{\textbf{arm-gcc}}\\
    \cline{2-13}
    
    \multicolumn{1}{c|}{}
    & \texttt{O0} & \texttt{O3} & \texttt{O0} & \texttt{O3}
    & \texttt{O0} & \texttt{O3} & \texttt{O0} & \texttt{O3}
    & \texttt{O0} & \texttt{O3} & \texttt{O0} & \texttt{O3}\\
    \hline
    ct\_select\_v1 & \ccmark{} & \circled{\cxmark{}} & \ccmark{} & \circled{\cxmark{}} & \ccmark{} & \cxmark{} & \ccmark{} & \ccmark{} & \ccmark{} & \ccmark{} & \ccmark & \ccmark \\
    ct\_select\_v2 & \ccmark{} & \cxmark{} & \ccmark{} & \cxmark{} & \ccmark{} & \cxmark{} & \ccmark{} & \ccmark{} & \ccmark{} & \ccmark{} & \ccmark & \ccmark \\
    ct\_select\_v3 & \ccmark{} & \ccmark{} & \ccmark{} & \cxmark{} & \ccmark{} & \cxmark{} & \ccmark{} & \ccmark{}  & \ccmark{} & \ccmark{} & \ccmark & \ccmark \\
    ct\_select\_v4 & \ccmark{} & \cxmark{} & \ccmark{} & \cxmark{} & \ccmark{} & \cxmark{} & \ccmark{} & \ccmark{} & \ccmark{} & \ccmark{} & \ccmark & \ccmark \\
    select\_naive (insecure) & \cxmark{} & \cxmark{} & \cxmark{} & \cxmark{} & \cxmark{} & \cxmark{} & \cxmark{} & \cxmark{} & \cxmark{} & \cxmark{} & \cxmark & \ccmark \\
    \hline
    \textbf{ct\_sort} & \ccmark{} & \cxmark{} & \ccmark{} & \cxmark{} & \ccmark{} & \ccmark{} & \cxmark{} & \ccmark{} & \cxmark{} & \ccmark{} & \cxmark{} & \ccmark{} \\
    \textbf{ct\_sort\_mult} & \ccmark{} & \cxmark{} & \ccmark{} & \cxmark{} & \ccmark{} & \ccmark{} & \cxmark{} & \ccmark{} & \cxmark{} & \ccmark{} & \cxmark{} & \ccmark{} \\
    \textbf{sort\_naive} (insecure) & \cxmark{} & \cxmark{} & \cxmark{} & \cxmark{} & \cxmark{} & \cxmark{} & \cxmark{} & \cxmark{} & \cxmark{} & \cxmark{} & \cxmark{} & \ccmark{} \\
    \hline
    \textbf{hacl\_utility} (\(\times 11\)) & \ccmark{} & \ccmark{} &
    \ccmark{} & \ccmark{} & \ccmark{} & \ccmark{} &
    \ccmark{} & \ccmark{} & \ccmark{} & \ccmark{} & \ccmark{} & \ccmark{} \\
    \textbf{openssl\_utility} (\(\times 13\)) & \ccmark{} & \ccmark{}
    & \ccmark{} & \ccmark{} & \ccmark{} & \ccmark{} &
    \ccmark{} & \ccmark{} & \ccmark{} & \ccmark{} & \ccmark{} & \ccmark{} \\
    \hline \textbf{tea\_encrypt} & \ccmark{} & \ccmark{} & \ccmark{} &
    \ccmark{} & \ccmark{} & \ccmark{} &
    \ccmark{} & \ccmark{} & \ccmark{} & \ccmark{} & \ccmark{} & \ccmark{} \\
    \textbf{tea\_decrypt} & \ccmark{} & \ccmark{} & \ccmark{} &
    \ccmark{} & \ccmark{} & \ccmark{} &
    \ccmark{} & \ccmark{} & \ccmark{} & \ccmark{} & \ccmark{} & \ccmark{} \\
    \hline
  \end{tabularx}
  \caption{Constant-time analysis of several functions compiled with
    \texttt{gcc} or \texttt{clang} (cl) and optimization levels
    \texttt{O0} or \texttt{03}. \ccmark{} indicate that the program is
    secure and \cxmark{} that it is insecure. \textbf{Bold programs}
    and \textbf{compilers} are extensions of~\cite{simon_what_2018}
    and \circled{\cxmark} indicates a different result
    than~\cite{simon_what_2018}.}\label{tab:verif_compilers}
\end{table}

We \emph{confirm} the main conclusion of Simon \emph{et
  al.}~\cite{simon_what_2018} that \texttt{clang} is more likely to
optimize away CT protections as the optimization level increases.
Yet, \emph{contrary to their work}, our experiments show that newer
versions of \texttt{clang} are not necessarily more likely than older
ones to break CT (e.g. \texttt{ct\_sort} is compiled to a
non-constant-time code with \texttt{clang-3.9} but not with
\texttt{clang-7.1}).

Surprisingly, in contrast with \texttt{clang}, \texttt{gcc}
optimizations tend to remove branches and thus, are less likely to
introduce vulnerabilities in constant-time code. Especially,
\texttt{gcc} \emph{for ARM produces secure binaries from the insecure
  source codes}~\footnote{The compiler takes advantage of the many ARM
  conditional instructions to remove conditional jumps.}
\texttt{sort\_naive} and \texttt{select\_naive}.

Although~\cite{simon_what_2018} reports that the
\texttt{ct\_select\_v1} function is secure in all their settings, we
find the opposite. Manual inspection confirms that \texttt{clang} with
\texttt{-O3} introduces a secret-dependent conditional jump violating
constant-time.

Finally, as previously discussed, we found that the \texttt{ct\_sort}
and \texttt{ct\_sort\_mult} functions, taken from the benchmark of the
\texttt{ct-verif}~\cite{almeida_verifying_2016} tool, can be compiled
to insecure binaries.  Those vulnerabilities are out of reach of
\texttt{ct-verif} because it targets LLVM code compiled with
\texttt{clang}, while the vulnerabilities are either introduced by
\texttt{gcc} or by \emph{backend passes} of \texttt{clang} --
we did confirm that \texttt{ct-verif} with the setting
\texttt{--clang-options="-O3 -m32 -march=i386"} does not report the
vulnerability.

\myparagraph{Conclusion (RQ1, RQ2).}  We perform an extensive analysis
over 338 samples of representative cryptographic primitive studied in
the literature
\cite{blazy_verifying_2017,zinzindohoue_hacl*:_2017,almeida_verifying_2016},
compiled with different versions and options of \texttt{clang} and
\texttt{gcc}, over x86 and ARM.  Overall, it demonstrates that
\brelse{} \modif{does scale}{does scale} to realistic applications for
both bug-finding and bounded-verification (RQ1), and that the
technology is generic (RQ2).
We also get the following interesting side results: 

\begin{itemize}
\item We proved CT-secure 296 binaries of interest; %
\item We found \modif{vulnerabilities in 3 programs}{3 new
    vulnerabilities} that slipped through previous analysis -- manual
  on binary code~\cite{simon_what_2018} or automated on
  LLVM~\cite{almeida_verifying_2016};
\item We significantly extend and automate a previous study on effects
  of compilers on CT~\cite{simon_what_2018}\modif{, going from 192
    configurations to 408}{};
\item We found that \texttt{gcc} optimizations tend to help enforcing
  CT -- on ARM, \texttt{gcc} even sometimes \modif{manages to
    produce}{produces} secure \modif{codes}{binaries} from insecure sources. 
\end{itemize}

\subsection{Comparisons  (RQ3,RQ4,RQ5)}
\label{sec:scalability}
We compare \brelse{} with standard approaches based on
self-composition (\textit{SC}) and relational symbolic execution
(\textit{RelSE}) (RQ3), then we analyze the performances of our
different simplifications (RQ4), and finally we investigate the
overhead of \brelse{} compared to standard SE, and whether our
simplifications are useful for SE (RQ5). %

Experiments are performed on the programs introduced in
\cref{sec:effectiveness} for bug-finding and
bounded-verification %
(338 samples, 70k instructions).
We report the following metrics: total number of unrolled instruction
\#I, number of instruction explored per seconds (\#I/s), total number
of queries sent to the solver (\#Q), number of exploration (resp.\
insecurity) queries (\(\text{\#Q}_{\text{e}}\)), (resp.\
\(\text{\#Q}_{\text{i}}\)), total execution time (T), timeouts
(\hourglass), programs proven secure (\ccmark), programs proven
insecure (\cxmark), unknown status (\csmark). Timeout is set to 3600
seconds.

\myparagraph{Comparison vs. Standard Approaches (RQ3).}
We evaluate \brelse{} against \textit{SC} and
\textit{RelSE}. %
Since no implementation of these methods fit our particular use-cases,
we implement them directly in \binsec{}. \textit{RelSE} is obtained by
disabling \brelse{} optimizations (\cref{sec:optims}), while
\textit{SC} is implemented on top of \textit{RelSE} by duplicating low
inputs instead of sharing them and adding the adequate preconditions.
Results are given in \cref{tab:scale_total_summary}.

\begin{table}[!htbp]
  \centering
  \setlength{\tabcolsep}{3pt}
\begin{tabularx}{\linewidth}{@{}|X|r|r|r|r|r|r|r|r|r|r|@{}}
  \cline{2-11}%
  \multicolumn{1}{c|}{} &
  \multicolumn{1}{c|}{\(\text{\#I}\)} &
  \multicolumn{1}{c|}{\#I/s} &
  \multicolumn{1}{c|}{\#Q} &
  \multicolumn{1}{c|}{\begin{tabular}{@{}c@{}}\(\text{\#Q}_{\text{e}}\)\end{tabular}} &
  \multicolumn{1}{c|}{\begin{tabular}{@{}c@{}}\(\text{\#Q}_{\text{i}}\)\end{tabular}} &
  \multicolumn{1}{c|}{T} &                                                      
  \multicolumn{1}{c|}{\hourglass} &
  \multicolumn{1}{c|}{\ccmark} &
  \multicolumn{1}{c|}{\cxmark} &
  \multicolumn{1}{c|}{\csmark} \\
  \hline
  \textit{SC}    &   252k &   3.9 & 170k &  16k & 154k & 65473 &  15 &  282 & 41 &  15 \\
  \textit{RelSE} &   320k &   5.4 &  97k &  19k &  78k & 59316 &  14 &  283 & 42 &  13 \\
  \hline
  \hline
  \brelse{}      &  22.8M &  3861 & 3.9k & 2.7k & 1.3k &  5895 &   0 &  296 & 42 &   0 \\
  \hline
\end{tabularx}

\caption{ \brelse{} vs.~standard approaches\label{tab:scale_total_summary}}
\end{table}

While \textit{RelSE} performs slightly better than \textit{SC}
($\times$1.38 speedup) thanks to a noticeable reduction of the number
of queries ($\approx$50\%), both techniques are not efficient enough
on binary code:
\textit{RelSE} times out in 14 cases and achieves an analysis speed of
only 5.4 instructions per second while \textit{SC} is
worse. %
\brelse{} \emph{completely outperforms both previous approaches},
especially its simplifications
drastically reduce the number of queries sent to the solver
(\(\times 60\) less insecurity queries than \textit{RelSE}):

\begin{itemize}
\item \brelse{} reports no timeout, it is \(715\) times faster than
  \textit{RelSE} and \(1000\) times faster than \textit{SC};

\item \brelse{} \modif{does perform}{performs} bounded-verification of
  large programs (e.g. \texttt{donna}, \texttt{des\_ct},
  \texttt{chacha20}, etc.)  \modif{}{that were} out of reach of
  standard methods.
\end{itemize}

\myparagraph{Performances of Simplifications
  (RQ4).}\label{sec:perfs-simpl} %
We consider on-the-fly read-over-write (\textit{FlyRow}), untainting
(\textit{Unt}) and fault-packing (\textit{fp}). Results are reported
in \cref{tab:scale_optims}:

\begin{table}[!tbp]
  \centering
  \setlength{\tabcolsep}{2pt}
\begin{tabular}{@{}|l|r|r|r|r|r|r|r|r|r|r|@{}}
  \hline
  \multicolumn{1}{|c|}{Version} &
  \multicolumn{1}{c|}{\(\text{\#I}\)} &
  \multicolumn{1}{c|}{\#I/s} &
  \multicolumn{1}{c|}{\#Q} &
  \multicolumn{1}{c|}{\begin{tabular}{@{}c@{}}\(\text{\#Q}_{\text{e}}\)\end{tabular}} &
  \multicolumn{1}{c|}{\begin{tabular}{@{}c@{}}\(\text{\#Q}_{\text{i}}\)\end{tabular}} &
  \multicolumn{1}{c|}{T} &                                                      
  \multicolumn{1}{c|}{\hourglass} &
  \multicolumn{1}{c|}{\ccmark} &
  \multicolumn{1}{c|}{\cxmark} &
  \multicolumn{1}{c|}{\csmark} \\
  \hline
  \hline
  \multicolumn{11}{|l|}{\textbf{Standard RelSE with \textit{Unt} and \textit{fp}}} \\
  \hline
  \textit{RelSE} &    320k &   5.4 &  96919 &  19058 & 77861 & 59316 & 14 & 283 & 42 & 13 \\
  + \textit{Unt} &    373k &   8.4 &  48071 &  20929 & 27142 & 44195 &  8 & 288 & 42 &  8 \\
  + \textit{fp}  &    391k &  10.5 &  33929 &  21649 & 12280 & 37372 &  7 & 289 & 42 &  7 \\
  \hline
  \hline
  \multicolumn{11}{|l|}{\textbf{\brelse{} (\textit{RelSE} + \textit{FlyRow} + \textit{Unt} + \textit{fp})}} \\
  \hline
  \textit{RelSE+FlyRow} &  22.8M & 3075 & 4018 &   2688 &    1330 &   7402 &    0 &     296 &        42 &         0 \\
  + \textit{Unt}        &  22.8M & 3078 & 4018 &   2688 &    1330 &   7395 &    0 &     296 &        42 &         0 \\
  + \textit{fp}         &  22.8M & 3861 & 3980 &   2688 &    1292 &   5895 &    0 &     296 &        42 &         0 \\
  \bottomrule
\end{tabular}

\caption{Performances of \brelse{} simplifications.\label{tab:scale_optims}}
\end{table}

\begin{itemize}
\item \textit{FlyRow} is the major source of improvement in \brelse{},
  drastically reducing the number of queries sent to the solver and
  allowing a \(\times 569\) speedup w.r.t.~\textit{RelSE};

\item Untainting and fault-packing do have a positive impact on
  \textit{RelSE} (untainting alone reduces the number of queries by
  50\%, the two optimizations together yield a \(\times 2\) speedup);

\item Yet, their impact is \modif{much}{} more modest once
  \textit{FlyRow} is activated: untainting \modif{leads only}{leads}
  to a very slight speedup, while fault-packing still achieves a
  \(\times 1.25\) speedup.
\end{itemize}

\noindent Still, \textit{fp} can be \modif{very}{} interesting on some
particular programs, \modif{if}{when the} precision of the bug report
is not the priority.
Consider for instance the non-constant-time version of \texttt{aes} in
BearSSL (i.e.\ \texttt{aes\_big}): \brelse{} without \textit{fp}
reports 32 vulnerable instructions in 1580 seconds, while \brelse{}
with \textit{fp} reports 2 vulnerable \emph{basic blocks} (covering
the 32 vulnerable instructions) in only 146 seconds.

\myparagraph{Comparison vs. Standard SE (RQ5).}\label{sec:comp-se} %
Standard SE is directly implemented in the \textsc{Rel} module and
models a single execution of the program with exploration queries {\it
  but without insecurity queries}. %
\modif{We also consider a recent SMT simplification technique,
  \textit{PostRow}~\cite{farinier_arrays_2018}, dedicated to array
  simplifications (formula preprocessing).}{We also consider a recent
  implementation of read-over-write~\cite{farinier_arrays_2018}
  implemented as a formula pre-processing, posterior to SE
  (\textit{PostRow}).}  %
Results are presented in \cref{tab:sse_relse_overhead2}.

\begin{table}[!htbp]
\setlength{\tabcolsep}{6pt}
  \centering
\begin{tabularx}{.9\linewidth}{@{}|Z|r|r|r|r|r|r|r|r|}
  \cline{2-6}%
  \multicolumn{1}{c|}{} &
  \multicolumn{1}{c|}{\#I} &
  \multicolumn{1}{c|}{\#I/s} &
  \multicolumn{1}{c|}{\#Q} &
  \multicolumn{1}{c|}{T} &
  \multicolumn{1}{c|}{\hourglass}\\
  \hline%
  \textit{SE}           &   440k &  15.1 &  23453 &  29122 &   7\\
  \textit{SE+PostRow}\cite{farinier_arrays_2018}   &   509k &  18.5 & 27252 &  27587 &   7\\
  \textit{SE+FlyRow}    &  22.8M &  6804 &  2688 &   3346 &   0\\
  \hline
  \textit{RelSE}        &   320k &   5.4 &  96919 &  59316 &  14\\
  \textit{RelSE+PostRow}&   254k &   4.0 &  75043 &  63693 &  16\\
  \brelse{}             &  22.8M &  3861 &   3980 &   5895 &   0\\
  \hline                
\end{tabularx}
\caption{Performances of relational symbolic execution compared to
  standard symbolic execution with/without binary level
  simplifications.\label{tab:sse_relse_overhead2}}
\end{table}

\begin{itemize}
\item The overhead of \brelse{} compared to our best setting for SE
  (\textit{SE+FlyRow}), in terms of speed (\#I/s), is only
  \(\times 1.8\). Hence CT comes almost for free on top of standard
  SE. This is consistent with the fact that our simplifications
  discard most insecurity queries, letting only the exploration
  queries which are also part of SE.

\item \textit{FlyRow} completely outperforms \textit{PostRow}. %
  First, \textit{PostRow} is not designed for relational verification
  and must reason about pairs of memory. Second, \textit{PostRow}
  simplifications are not propagated along the execution and must be
  recomputed for every query, producing a significant
  simplification-time overhead. On the contrary, \textit{FlyRow}
  models a single memory containing relational values and propagates
  along the symbolic execution.

\item \textit{FlyRow} also improves the performance of standard SE by
  a factor \(450\), performing much better than
  \modif{state-of-the-art binary-level simplifications applied as a
    post-processing of the formula}{\textit{PostRow}} in our
  experiments.
\end{itemize}

\myparagraph{Conclusion (RQ3, RQ4, RQ5).}  \brelse{} performs
significantly better than previous approaches to relational symbolic
execution (\(\times 715\) speedup vs.~\textit{RelSE}). The very main
source of improvement is the \textit{FlyRow} on-the-fly simplification
(\(\times 569\) speedup vs.~\textit{RelSE}, $\times 60$ less
insecurity queries).
Note that, in our context, \textit{FlyRow} completely outperforms
state-of-the-art \modif{SMT formula simplifications for low-level
  constraints}{binary-level simplifications}, as they are not designed
to efficiently cope with relational properties and introduce a
significant simplification-overhead at every query.
Fault-packing and untainting, while effective over \textit{RelSE},
have a much slighter impact once \textit{FlyRow} is activated;
\modif{they can still be useful in certain situations or in the case
  of fault packing when report precision is not the main
  concern.}{fault-packing can still be useful when report precision is
  not the main concern.}
Finally, \modif{here}{in our experiments}, \textit{FlyRow}
significantly improves the performance of standard SE (\(\times 450\)
speedup).

\section{Discussion}\label{sec:discussion}

\myparagraph{Implementation limitations.} %
Our implementation shows three main limitations commonly found in
research prototypes: %
it does not support dynamic libraries -- executable must be statically
linked or stubs must be provided for external function calls, it does
not implement predefined syscall stubs, and it does not support
floating point instructions. These problems are orthogonal to the core
contribution of this paper and the two first ones are essentially
engineering tasks.
Moreover, the prototype is already efficient on real-world case
studies.

\myparagraph{Threats to validity in experimental evaluation.} %
We assessed the effectiveness of our tool on several known secure and
insecure real-world cryptographic binaries, many of them taken from
prior studies. All results have been crosschecked with the expected
output, and manually reviewed in case of deviation.

Our prototype is implemented as part of
\binsec{}~\cite{david_binsec/se:_2016}, whose efficiency and
robustness have been demonstrated in prior large scale studies on both
adversarial code and managed
code~\cite{bardin_backward-bounded_2017,recoules_ase_2019,DBLP:conf/cav/FarinierBBP18,david_specification_2016}. The
IR lifting part has been positively evaluated in external
studies~\cite{DBLP:journals/tocs/ChipounovKC12,kaist-bar-workshop} and
the symbolic engine features aggressive formula
optimizations~\cite{farinier_arrays_2018}.  All our experiments use
the same search heuristics (depth-first) and, for
bounded-verification, smarter heuristics do not change the
performances. %
Also, we tried Z3 and confirmed the better performance of Boolector.

Finally, we compare our tool to our own versions of \textit{SC} and
\textit{RelSE}, %
primarily because none of the existing tools can be easily adapted for
our setting, and also because this allows comparing very close
implementations.

\section{Related Work}\label{sec:related} %
Related work has already been lengthly discussed along the paper. We
add here only a few additional discussions, as well as an overview of
existing SE-based tools for information flow
(\cref{tab:comparison_se}) and an overview of (other) existing
automatic analyzers for CT (\cref{tab:comparison_ct}), partly taken
from~\cite{almeida_verifying_2016}.

\begin{table}[htbp!]
  \centering
  \footnotesize
  \setlength\tabcolsep{1.5pt} %
  \begin{tabularx}{\linewidth}{|X|c|c|c|c|r|r|}
    \hline
    Tool & Target & NI & Technique & P/BV/BF/C & \(\approx\!\)XP max & \(\text{I}_{u}/s\) \\
    \hline
    RelSym~\cite{farina_relational_2017} &
    imp-for & \ccmark & RelSE & \ccmark/\ccmark/\ccmark/\ccmark & 10 loc & \textsc{na} \\
    IF-exploit\cite{do_exploit_2015} &
    Java & \ccmark & SC & \ccmark/\ccmark/\ccmark/\ccmark & 20 loc & \textsc{na} \\
    Type-SC-SE\cite{milushev_noninterference_2012} &
    C & \ccmark & type-based SC & \cxmark/\cxmark/\ccmark/\ccmark & 20 loc & \textsc{na} \\
    Casym~\cite{brotzman_casym:_2019} &
    LLVM & \ccmark & SC+over-approx & \ccmark/\ccmark/\cxmark/\cxmark & 200 (C) & \textsc{na} \\ 
    \hline
    IF-low-level~\cite{balliu_automating_2014} &
    \textbf{binary} & \ccmark & SC + invariants & \ccmark/\ccmark/\cxmark/\cxmark & 250 \(\text{I}_{s}\)& \textsc{na}\\
    IF-firmware\cite{subramanyan_verifying_2016} &
    \textbf{binary} & \cxmark & SC + concretize & \cxmark/\cxmark/\ccmark/\ccmark & 500k \(\text{I}_{u}\) & 260 \\
    CacheD~\cite{wang_cached:_2017} &
    \textbf{binary} & \cxmark & concret+tainting & \cxmark/\cxmark/\ccmark/\ccmark & 31M \(\text{I}_{u}\) & 2010 \\
    \hline
    \hline
    \brelse{} &
                \textbf{binary} & \cxmark & RelSE + simpl. &
                \cxmark/\ccmark/\ccmark/\ccmark & 10M \(\text{I}_{u}\) & 3861  \\    
    \hline
  \end{tabularx}
  \caption{SE-based tools for Information Flow. NI indicates whether the
    technique handles  general non-interference (diverging paths) or
    not (CT-like properties), P: proof, BV: 
    bounded-verification, BF: bug-finding, C: counterexample. %
    \(\text{I}_{s}\): static instr.,  \(\text{I}_{u}\): unrolled instr., 
    \textsc{na}: non-applicable. }
  \label{tab:comparison_se}
\end{table}

\myparagraph{Self-compositon and SE} %
has first been used by Milushev \emph{et
  al.}~\cite{milushev_noninterference_2012}. They use type-directed
self-composition and dynamic symbolic execution to find bugs of
\emph{noninterference} but they do not address scalability and their
experiments are limited to toy examples. The main issues here are the
quadratic explosion of the search space (due to the necessity of
considering diverging paths) and the complexity of the underlying
formulas.
Later works~\cite{do_exploit_2015,balliu_automating_2014} suffer from
the same problems.

\emph{Instead of considering the general case of noninterference, we
  focus on CT, and we show that it remains tractable for SE with
  adequate optimizations.
}

\myparagraph{Relational symbolic execution.} %
\emph{Shadow Symbolic
  Execution}~\cite{cadar_shadow_2014,palikareva_shadow_2016} aims at
efficiently testing evolving softwares by focusing on the new
behaviors introduced by a patch.
The paper introduces the idea of \emph{sharing formulas} across two
executions in the same SE instance. The term \emph{relational symbolic
  execution} has been coined \modif{in a recent paper}{more
  recently}~\cite{farina_relational_2017} but this work is limited to
a simple toy imperative language and do not address scalability.

\emph{We maximize sharing between pairs of executions, as ShadowSE
  does, but we also develop specific optimizations tailored to the
  case of binary-level CT. Experiments show that our optimizations are
  crucial in this context. }

\myparagraph{Symbolic execution for CT.} %
Only three previous works in this category achieve scalability, yet at
the cost of either precision or soundness.
Wang {et al.}~\cite{wang_cached:_2017} and Subramanyan \emph{et
  al.}~\cite{subramanyan_verifying_2016} sacrifice soundness for
scalability (no bounded-verification).  The former performs symbolic
execution on fully concrete traces and only symbolize the secrets.
The latter concretizes memory accesses.
In both cases, they may miss feasible paths as well as
vulnerabilities.
Brotzman \emph{et al.}~\cite{brotzman_casym:_2019} take the opposite
side and sacrifice precision for scalability (no bug-finding). %
Their analysis scales by over-approximating loops and resetting the
symbolic state at chosen code locations.

We adopt a different approach and scale by heavy formula
optimizations, allowing us to keep both correct bug-finding (BF) and
correct bounded-verification (BV). Interestingly, our method is faster
than these approximated ones.  %
\emph{We propose the first technique for CT-verification at
  binary-level that is correct for BF and BV and scales on real world
  cryptographic examples.}
Moreover, our technique is compatible with the previous approximations
for extra-scaling.

\myparagraph{Other Methods for CT Analysis.} 
\emph{Static approaches} %
based on sound static
analysis~\cite{agat_transforming_2000,won_program_2006,barthe_system-level_2014,bacelar_almeida_formal_2013,blazy_verifying_2017,kopf_automatic_2012,doychev_cacheaudit:_2015,doychev_rigorous_2017,almeida_verifying_2016,DBLP:conf/cc/RodriguesPA16}
give formal guarantees that a program is free from time side-channels
but they cannot find bugs when a program is rejected.
Some works also propose program transformations to make a program
secure~\cite{agat_transforming_2000,won_program_2006,chattopadhyay_symbolic_2018,wu_eliminating_2018,brotzman_casym:_2019}
but they consider less capable attackers and target higher-level code.
\emph{Dynamic approaches} for constant-time are precise (they find real
violations) but limited to a subset of the 
execution traces,  hence they are not complete. These
techniques include statistical analysis~\cite{reparaz_dude_2017},
dynamic binary
instrumentation~\cite{langley_imperialviolet_2010,wichelmann_microwalk:_2018},
and dynamic symbolic execution
(DSE)~\cite{chattopadhyay_quantifying_2017}.

\begin{table}[tbp!]
  \centering
  \footnotesize
  \setlength\tabcolsep{2pt} %
  \begin{tabularx}{\linewidth}{|X|c|c|c|c|}
    \hline
    Tool & Target & Analysis & Technique & P/BV/BF/C \\
    \hline
    ct-ai~\cite{blazy_verifying_2017} &
        C & static & abstract-interpretation & \ccmark/\ccmark/\cxmark/\cxmark \\
    FlowTracker~\cite{DBLP:conf/cc/RodriguesPA16} &
        LLVM & static & type-system & \ccmark/\ccmark/\cxmark/\cxmark \\
    ct-verif~\cite{almeida_verifying_2016} &
        LLVM & static & logical, product-programs & \ccmark/\ccmark/\cxmark\(^{*}\)/ \cxmark \\   %
    Casym~\cite{brotzman_casym:_2019}\(^{\ddagger}\) &
        LLVM & static & over-approx.~SE & \ccmark/\ccmark/\cxmark/\cxmark \\   %
    VirtualCert\(^\dag\)~\cite{barthe_system-level_2014} & %
        x86 & static & type-system & \ccmark/\ccmark/\cxmark/\cxmark \\
    \hline %
    ctgrind~\cite{langley_imperialviolet_2010} & %
        \textbf{binary} & dynamic & Valgrind & \cxmark/\cxmark/\cxmark/\ccmark \\ 
    CacheAudit~\cite{doychev_rigorous_2017}\(^{\ddagger}\) &
        \textbf{binary} & static & abstract-interpretation & \ccmark/\ccmark/\cxmark/\cxmark \\
    CacheD~\cite{wang_cached:_2017}\(^{\ddagger}\) &
        \textbf{binary} & dynamic & DSE & \cxmark/\cxmark/\ccmark/\ccmark \\
    \hline
    \hline
    \brelse{} &
        \textbf{binary} & SE & RelSE + simpl. & \cxmark/\ccmark/\ccmark/\ccmark \\
    \hline
  \end{tabularx}
  \caption{Automatic analysis tools for CT-like properties
    (see~\cite{almeida_verifying_2016}). 
\(^{*}\)~ct-verif  can
    be incomplete because of invariant inference. %
    \(^\dag\)~As part of CompCert, cannot be used on arbitrary
    executables. \(^{\ddagger}\)~Also implements a cache model.}
  \label{tab:comparison_ct}
\end{table}

\section{Conclusion}\label{sec:conclusion} %
We tackle the problem of designing an automatic and efficient
binary-level analyzer for \emph{constant-time}, enabling both
bug-finding and bounded-verification on real-world cryptographic
implementations.  Our approach is based on \emph{relational symbolic
  execution} together with original \emph{dedicated optimizations}
reducing the overhead of relational reasoning and allowing for a
significant speedup.
Our prototype, \brelse{}, is shown to be highly efficient compared to
alternative approaches. We used it to perform extensive binary-level
CT analysis for a wide range of cryptographic implementations and to
automate and extend a previous study of CT preservation by compilers.
We found three vulnerabilities that slipped through previous manual
and automated analyses, and we discovered that \texttt{gcc -O0} and
backend passes of \texttt{clang} introduce violations of CT out of
reach of state-of-the-art CT verification tools at LLVM or source
level.

\newpage

\printbibliography{}

\appendix
\subsection{Full Set of rules}
\subsubsection{Concrete Evaluation}\label{app:concrete_evaluation}
The full set of rules for the concrete evaluation is reported in
\Cref{fig:dba_semantics_full}.

\begin{figure}[htbp]
\centering
\footnotesize

\fbox{\begin{minipage}{.98\linewidth}
  \begin{mathpar}

    \boxed{\textbf{Expr}} %
    \hspace{2em}

    \inferrule*[left={cst}]{ }{\ceconf{\cregmap}{\cmem}{\texttt{bv}} \eeval{} \texttt{bv}}
    \quad

    \inferrule*[left={var}]{ }{\ceconf{\cregmap}{\cmem}{\texttt{v}}
      \eeval{} \cregmap\ \texttt{v}} \and

    \inferrule*[left={unop}]{
      \ceconf{\cregmap}{\cmem}{e} \eeval{} \texttt{bv}%
    }{
      \ceconf{\cregmap}{\cmem}{\blackdiamond_{u} e} \eeval{} \blackdiamond_{u} \texttt{bv}
    }

    \inferrule*[left={binop}]{
      \ceconf{\cregmap}{\cmem}{e_1} \eeval{} \texttt{bv}_{\texttt{1}}\\
      \ceconf{\cregmap}{\cmem}{e_2} \eeval{} \texttt{bv}_{\texttt{2}}%
    }{
      \ceconf{\cregmap}{\cmem}{e_1 \blackdiamond_{b} e_2} \eeval{}
      \texttt{bv}_{\texttt{1}} \diamond_{b} \texttt{bv}_{\texttt{2}}
    }

    \inferrule*[left={load}]
    {
      \ceconf{\cregmap}{\cmem}{e} \ceeval{\leakvar} \texttt{bv} \\
    }
    {
      \ceconf{\cregmap}{\cmem}{\load{} e} %
      \ceeval{\leakvar \cdot [\texttt{bv}] } \cmem~\texttt{bv} 
    }\and
  \end{mathpar}\end{minipage}}

  \fbox{\begin{minipage}{.98\linewidth}\begin{mathpar}
     \boxed{\textbf{Instr}}

    \inferrule*[left={s\_jump}]{
      \locmap{\locvar} = \texttt{goto}\ \locvar'
    }{
      \cconf{\locvar}{\cregmap}{\cmem} \cleval{[\locvar]}
      \cconf{\locvar'}{\cregmap}{\cmem}
    }\and

    \inferrule*[lab={d\_jump}]{
      \locmap{l} = \texttt{goto}\ e \\
      \ceconf{\cregmap}{\cmem}{e} \ceeval{\leakvar} \texttt{bv} \\
      \locvar' \mydef to\_loc(\texttt{bv})
    }
    {
      \cconf{\locvar}{\cregmap}{\cmem} %
      \cleval{\leakvar \cdot [\locvar']} %
      \cconf{\locvar'}{\cregmap}{\cmem}
    }\and

    \inferrule*[lab={ite-true}]{
      \locmap{l} = \texttt{ite}\ e\ \texttt{?}\ \locvar_1 \texttt{:}\ \locvar_2\\  
      \ceconf{\cregmap}{\cmem}{e} \ceeval{\leakvar} \texttt{bv} \\
      \texttt{bv} \neq 0
    }
    {
      \cconf{\locvar}{\cregmap}{\cmem} %
      \cleval{\leakvar \cdot [\locvar{}_1]} %
      \cconf{\locvar{}_1}{\cregmap}{\cmem}
    }\and

    \inferrule*[lab={ite-false}]{
      \locmap{l} = \texttt{ite}\ e\ \texttt{?}\ \locvar_1 \texttt{:}\ \locvar_2\\  
      \ceconf{\cregmap}{\cmem}{e} \ceeval{\leakvar} \texttt{bv} \\
      \texttt{bv} = 0
    }
    {
      \cconf{\locvar}{\cregmap}{\cmem} %
      \cleval{\leakvar \cdot [\locvar{}_2]} %
      \cconf{\locvar{}_2}{\cregmap}{\cmem}
    }\and

    \inferrule*[left={assign}]{
      \locmap{l} = \texttt{v := } e\\
      \cconf{\locvar}{\cregmap}{\cmem}{e} \ceeval{t} \texttt{bv}
    }{
      \cconf{\locvar}{\cregmap}{\cmem} \cleval{t}
      \cconf{\locvar+1}{\cregmap[\texttt{v} \mapsto{} \texttt{bv}]}{\cmem}
    }\and

    \inferrule*[lab={store}]{
      \locmap{l} = \store{} e := e'\\
      \ceconf{\cregmap}{\cmem}{e} \ceeval{\leakvar} \texttt{bv} \\
      \ceconf{\cregmap}{\cmem}{e'} \ceeval{\leakvar'} \texttt{bv}' \\
    }
    { %
      \cconf{\locvar}{\cregmap}{\cmem} %
      \cleval{\leakvar' \cdot \leakvar \cdot [\texttt{bv}]} %
      \cconf{\locvar+1}{\cregmap}{\cmem[\texttt{bv} \mapsto \texttt{bv'}]} %
    }
  \end{mathpar}\end{minipage}}
\caption{Concrete evaluation of DBA instructions and expressions,
  where \(\cdot\) is the concatenation of leakages and
  \(to\_loc : \bvset{32} \to \locset\) converts a bitvector to a
  location.}
  \label{fig:dba_semantics_full}
\end{figure}

\subsubsection{Symbolic Evaluation}\label{app:symbolic_evaluation}
The full set of rules for the symbolic evaluation is reported in
\Cref{fig:eval_instr_sha_full}.

\begin{figure}[htbp]
\footnotesize

\fbox{\begin{minipage}{.98\linewidth}
  \begin{mathpar}

    \boxed{\textbf{Expr}} \hspace{2em}
    \inferrule*[left={cst}]{ }{\econfold{\regmap}{\smem}{\texttt{bv}} \eeval{} \simple{bv}}
    \quad
    \inferrule*[left={var}]{ }{\econfold{\regmap}{\smem}{\texttt{v}}
      \eeval{} \regmap\ \texttt{v}} \and

    \inferrule*[left={unop}]{
      \econfold{\regmap}{\smem}{e} \eeval{} \rel{\phi} \\
      \rel{\varphi} \mydef{} \diamond_{u} \rel{\phi}
    }{
      \econfold{\regmap}{\smem}{\blackdiamond_{u} e} \eeval{} \rel{\varphi}
    }

    \inferrule*[left={binop}]{
      \econfold{\regmap}{\smem}{e_1} \eeval{} \rel{\phi}\\
      \econfold{\regmap}{\smem}{e_2} \eeval{} \rel{\psi}\\
      \rel{\varphi} \mydef{} \rel{\phi{}} \diamond_{b} \rel{\psi}
    }{
      \econfold{\regmap}{\smem}{e_1 \blackdiamond_{b} e_2} \eeval{} \rel{\varphi}
    }

    \inferrule*[left={load}]{
      \econfold{\regmap}{\smem}{e} \eeval{} \rel{\phi}\\
     \rel{\varphi} \mydef{} \pair{select(\lproj{\smem}, \lproj{\rel{\phi}}) }{
                                    select(\rproj{\smem}, \rproj{\rel{\phi}})} \\
                                    \secleak(\rel{\phi})
    }{
      \econfold{\regmap}{\smem}{\load{} e} \eeval{} \rel{\varphi}
    }\and
  \end{mathpar}\end{minipage}}

  \fbox{\begin{minipage}{.98\linewidth}\begin{mathpar}
    \boxed{\textbf{Instr}} \hspace{5em}

    \inferrule*[left={s\_jump}]{
      \locmap{l} = \texttt{goto}\ l'
    }{
      \iconfold{l}{\regmap}{\smem}{\pc{}} \ieval{}
      \iconfold{l'}{\regmap}{\smem}{\pc{}}
    }\and

    \inferrule*[left={d\_jump}]{
      \locmap{l} = \texttt{goto}\ e\\
      \econfold{\regmap}{\smem}{e}  \eeval{} \rel{\varphi}  \\
      \pi{}' \mydef{} \pi{} \wedge{} %
      (\lproj{\rel{\varphi{}}} = \rproj{\rel{\varphi{}}})\\
      M\sat{\pi{}' }  \\ l' \mydef M(\lproj{\rel{\varphi{}}}) \\
      \secleak(\rel{\varphi})
    }{
      \iconfold{l}{\regmap}{\smem}{\pc} \ieval{}
      \iconfold{l'}{\regmap}{\smem}{%
       \pi{}'}
    }\and

    \inferrule*[lab={ite-true}]{
      \locmap{l} = \texttt{ite}\ e\ \texttt{?}\ l_{true} \texttt{:}\ l_{false}\\
      l' \mydef  l_{true} \\
      \econfold{\regmap}{\smem}{e } \eeval{} \rel{\varphi} \\
      \pc{}' \mydef{} \pc{} \wedge{} %
      (true= \lproj{\rel{\varphi{}}} = \rproj{\rel{\varphi{}}})\\
         \sat{\pc'} \\
                \secleak(\rel{\varphi})
   }{
      \iconfold{l}{\regmap}{\smem}{\pc{}} \ieval{}
      \iconfold{l'}{\regmap}{\smem}{\pc{}'}
    }\and

    \inferrule*[lab={ite-false}]{
      \locmap{l} = \texttt{ite}\ e\ \texttt{?}\ l_{true} \texttt{:}\ l_{false}\\
      l' \mydef  l_{false} \\
      \econfold{\regmap}{\smem}{e } \eeval{} \rel{\varphi} \\
      \pc{}' \mydef{} \pc{} \wedge{} %
      (false = \lproj{\rel{\varphi{}}} = \rproj{\rel{\varphi{}}})\\
         \sat{\pc'} \\
                \secleak(\rel{\varphi})
   }{
      \iconfold{l}{\regmap}{\smem}{\pc{}} \ieval{}
      \iconfold{l'}{\regmap}{\smem}{\pc{}'}
    }\and  

    \inferrule*[left={assign}]{
      \locmap{l} = v := e\\
      \econfold{\regmap}{\smem}{e} \eeval{} \rel{\varphi} \\
      \rel{\varphi}' \mydef{} \canonical(\rel{\varphi}) \\
      \regmap' \mydef{} \regmap[v \mapsto{} \rel{\varphi}'] \\
      \pc' \mydef{} \pc{} \wedge{} (\lproj{\rel{\varphi}}' = \lproj{\rel{\varphi}})
      \wedge{} (\rproj{\rel{\varphi}}' = \rproj{\rel{\varphi}})
    }{
      \iconfold{l}{\regmap}{\smem}{\pc{}} \ieval{}
      \iconfold{l+1}{\regmap'}{\smem}{\pc'}
    }\and

    \inferrule*[lab={store}]{
      \locmap{l} = \store{} e := e' \\
      l' = l+1\\
      \econfold{\regmap}{\smem}{e}  \eeval{} \rel{\varphi} \\
      \econfold{\regmap}{\smem}{e'} \eeval{} \rel{\phi} \\
    \smem' \mydef{}  \pair{ store(\lproj{\smem},\lproj{\rel{\varphi}},\lproj{\rel{\phi}})}{store(\rproj{\smem},\rproj{\rel{\varphi}},\rproj{\rel{\phi}})} \\
       \pc' \mydef{} \pc{} \wedge \lproj{\smem}' = store(\lproj{\smem},\lproj{\rel{\varphi}},\lproj{\rel{\phi}})
                          \wedge \rproj{\smem}' = store(\rproj{\smem},\rproj{\rel{\varphi}},\rproj{\rel{\phi}})\\
                                 \secleak(\rel{\varphi})
    }{
      \iconfold{l}{\regmap}{\smem}{\pc{}} \ieval{}
      \iconfold{l'}{\regmap}{\smem'}{\pc{}'}
    }
  \end{mathpar}\end{minipage}}
\caption{Symbolic evaluation of DBA instructions and expressions where
  \(\canonical(\protect\rel{\varphi})\) returns
  \(\protect\rel{\varphi}\) if it is in canonical form or a temporary
  variable otherwise; and \(\diamond_{u}\) (resp. \(\diamond_{b}\)) is
  the logical counterpart of the concrete operator
  \(\blackdiamond_{u}\) (resp. \(\blackdiamond_{b}\)).}
\label{fig:eval_instr_sha_full}
\end{figure}

\subsection{Proofs}
\noindent
We recall essential elements that we will use in the proofs.

\begin{proposition}\label{hyp:deterministic}
  Concrete semantics is deterministic, c.f.\ rules of the concrete
  semantics in \cref{fig:dba_semantics_full}.
\end{proposition}

\begin{proposition}\label{hyp:ct_upk}
  If a program \(\prog{}\) is constant-time up to \(k\) then then for
  all \(j \leq k\), \(\prog{}\) is constant-time up to \(j\).
\end{proposition}

\begin{hypothesis}\label{hyp:stuck}
  Symbolic execution does not get stuck unless
  \($\secleak$\) evaluates to false. In particular, this implies that
  \(\prog{}\) is defined on all locations computed during the symbolic
  execution.
\end{hypothesis}

\subsubsection{Proof of Relative Completeness of RelSE
  (\Cref{thm:completeness})}\label{app:completeness}
\completeness*

\begin{proof} \emph{(Induction on k).}  Case \(k = 0\) is trivial.
  
  Let \(c_k\) and \(c_{k}'\) be concrete configurations and \(s_k\) a
  symbolic configuration for which the inductive hypothesis holds up
  to \(k-1\). We need to show that \cref{thm:completeness} still holds
  at step \(k\), meaning that if \(\prog\) is constant-time up to step
  \(k\), then for each concrete steps
  \(c_{k-1} \cleval{\leakvar} c_{k}\) and
  \(c_{k-1}' \cleval{\leakvar'} c_{k}'\) such that
  \(\leakvar = \leakvar'\), we need to show that we can perform a step
  in the symbolic execution \(s_{k-1} \ieval{} s_{k}\) and that
  \({c_{k} \concsym{l}{M} s_{k}} ~\wedge~ {c_{k}' \concsym{r}{M}
    s_{k}}\) holds.

  We can proceed case by case on the concrete evaluation of
  \(c_{k-1}\) and \(c_{k-1}'\).

  \textbf{Case \rulename{store}}: In the concrete execution, the
  instruction \(\store{} e_{idx} := e_{val}\) is evaluated and leaks
  the index \(e_{idx}\) of the store. %
  By \Cref{hyp:deterministic}, concrete semantics of the
  \rulename{store} rule, and because \(t = t'\), we have:
  \begin{equation}
    \label{eq:1}
    c_{k-1}\ e_{idx} \ceeval{} \mathtt{bv_{i}} \text{ and }
    c_{k-1}'\ e_{idx} \ceeval{} \mathtt{bv_{i}}
  \end{equation}
  
  \parsquare{} First, we show that there exists a step from
  \(s_{k-1} \mydef \iconfold{\locvar}{\regmap}{\smem}{\pc}\) in the
  symbolic execution. We have to consider the case where the
  symbolic evaluation of expressions is not stuck and the case where
  the evaluation of an expression gets stuck.

  \textbf{Case} evaluation of expressions is not stuck. Let \(\iota\)
  be the symbolic index such that
  \(\econfold{\regmap}{\smem}{e_{idx}} \eeval{} \rel{\iota}\).  %
  To apply the rule \rulename{store}, we must ensure that
  \(\secleak(\rel{\iota})\) holds. If \(\rel{\iota} = \simple{\iota}\)
  then \(\secleak\) is true, otherwise we must show
  \(\unsat \pc \wedge \lproj{\rel{\iota}} \neq
  \rproj{\rel{\iota}}\). We show that by contradiction.

  \textbf{Assume} that there exists \(M'\) such that
  \(M' \sat \pc \wedge \lproj{\rel{\iota}} \neq \rproj{\rel{\iota}}\)
  and let \(M'(\lproj{\rel{\iota}}) = \mathtt{bv_i}\) and
  \(M'(\rproj{\rel{\iota}}) = \mathtt{bv_i'}\). Notice that
  \(\mathtt{bv_i} \neq \mathtt{bv_i'}\).

  Let \(d_0, d_0', d_{k-1}, d_{k-1}'\) be concrete configurations such
  that \(d_0 \concsym{l}{M'} s_0\), \(d_0' \concsym{r}{M'} s_0\),
  \(d_{k-1} \concsym{l}{M'} s_{k-1}\), and
  \(d_{k-1}' \concsym{r}{M'} s_{k-1}\). %
  From \cref{thm:correctness}, there are concrete executions
  \(d_0 \cleval{\leakvar}^{k-1} d_{k-1}\) and
  \(d_0' \cleval{\leakvar'}^{k-1} d_{k-1}'\). Because
  \(d_k \concsym{l}{M'} s_{k-1}\) and
  \(d_{k-1}' \concsym{r}{M'} s_{k-1}\), we know that \(d_{k-1}\) and
  \(d_{k-1}'\) also evaluate the instruction \(\store{} e_{idx} := e_{val}\), and
  that
  \[d_{k-1}\ e_{idx} \ceeval{} M'(\lproj{\rel{\iota}}) \text{ and } %
    d_{k-1}'\ e_{idx} \ceeval{} M'(\rproj{\rel{\iota}})\] %
  Therefore, we have
  \(d_0 \cleval{\leakvar \cdot [\mathtt{bv_i}]}^{k} d_{k}\) and
  \(d_0' \cleval{\leakvar' \cdot [\mathtt{bv_i'}]}^{k} d_{k}'\) with
  \(\mathtt{bv} \neq \mathtt{bv_i'}\), meaning that \(\prog\) is not
  constant-time at step \(k\). This contradicts the hypothesis that
  \(\prog\) is constant-time up to \(k\), hence
  \(\unsat \pc \wedge \lproj{\rel{\iota}} \neq \rproj{\rel{\iota}}\)
  and \(\secleak\) holds. Therefore, there exists a step
  \(s_{k-1} \ieval{} s_{k}\) in the symbolic execution.

  \textbf{Case} \(\econfold{\regmap}{\smem}{e_{idx}}\) is stuck. From
  \Cref{hyp:stuck}, \(\secleak{}\) evaluates to false in the
  evaluation of the expression. With the same reasoning as in the
  previous case, this implies that \(\prog{}\) is not constant-time at
  step \(k\), which contradicts the hypothesis of
  \cref{thm:completeness}. Hence the symbolic evaluation of \(e_{idx}\)
  cannot be stuck.

  \textbf{Case} \(\econfold{\regmap}{\smem}{e_{val}}\) is stuck is
  analogous.
  
  \parsquare{} Now, we need to show that there is a model \(M'\) such
  that
  \(c_{k} \concsym{l}{M'} s_{k} ~\wedge~ c_{k}' \concsym{r}{M'}
  s_{k}\) holds.

  Recall that the program location in \(s_{k-1}\) evaluates to a store
  instruction \(\store{} e_{idx} := e_{val}\) and let \(\rel{\iota}\)
  be the symbolic index and \(\rel{\nu}\) be the symbolic value,
  meaning that
  \(\econfold{\regmap}{\smem}{e_{idx}} \eeval{} \rel{\iota} \) and
  \(\econfold{\regmap}{\smem}{e_{val}} \eeval{} \rel{\nu}\). Let
  \(s_k \mydef \iconfold{\locvar_k}{\regmap_k}{\smem_k}{\pc_k}\).

  First we need to show that the location in configurations \(c_k\),
  \(c_k'\) and \(s_k\) are identical. Because concrete and symbolic
  \rulename{store} rules increment the program location by 1 and
  because the program locations are identical in \(c_{k-1}\),
  \(c_{k-1}'\) and \(s_{k-1}\) (from induction hypothesis), the
  program locations are still identical in \(c_k\), \(c_k'\) and
  \(s_k\).

  Second, we have to show that there exists \(M'\) such that
  \(M' \sat \pc_k\) and that for all expression \(e\), either the
  symbolic evaluation gets stuck on \(e\), or
  \(\econfold{\regmap_{k}}{\smem_{k}}{e} \eeval{} \rel{\varphi}\) and
  \begin{equation}\label{eq:concsym_holds}
    \begin{aligned}
      M'(\lproj{\rel{\varphi}}) = \mathtt{bv} \iff c_k~e \ceeval{}
      \mathtt{bv} ~\wedge\\ %
      M'(\rproj{\rel{\varphi}}) = \mathtt{bv} \iff c_k'~e \ceeval{}
      \mathtt{bv}
    \end{aligned}
  \end{equation}

  \noindent
  We can build the new model \(M'\) as
  \begin{multline*}
    M' \mydef M[\smem_k \mapsto \pair{m_l}{m_r}]\\
    \begin{aligned}
      \text{where } &m_l \mydef M(\lproj{\smem})[M(\lproj{\rel{\iota}}) \mapsto M(\lproj{\rel{\nu}})]\\
      \text{and }   &m_r \mydef M(\rproj{\smem})[M(\rproj{\rel{\iota}}) \mapsto M(\rproj{\rel{\nu}})]      
    \end{aligned}
  \end{multline*}
  Intuitively, \(M'\) is equal to the old model \(M\) in which we add
  the new symbolic memory \(\smem_k\), mapping to the concrete value
  of the old memory \(M(\smem)\) where the index \(M(\rel{\iota})\)
  maps to the value \(M(\rel{\nu})\). Notice that \(M' \sat \pc_k\).

  \textbf{Case} of the left projection (right case is analogous).  We
  can prove by induction on the structure of expressions that for any
  expression, if \(s_k\) does not get stuck then
  \cref{eq:concsym_holds} holds for the model \(M'\). %
  Note that only the memory is updated from step \(k-1\) to step
  \(k\), meaning that \(c_k\), \(s_k\), and \(M'\) only differ from
  \(c_{k-1}\), \(s_{k-1}\) and \(M\) on expressions involving the
  memory. Thus, we only need to consider the rule \rulename{load}, as
  the proof for other rules directly follows from the induction
  hypothesis and the definition of \(M'\).
  
    Assume an expression \(\load e\) such that \(s_k\) does not get
  stuck and let
  \(\econfold{\regmap_{k}}{\smem_{k}}{e} \eeval{} \rel{\iota'}\) end
  \(\econfold{\regmap_{k}}{\smem_{k}}{\load e} \eeval{} \rel{\nu'}\).
  We show that if \cref{eq:concsym_holds} holds for the expression
  \(e\), then it holds for the expression \(\load e\). Formally, we
  must show that if %
  \({M'(\lproj{\rel{\iota'}}) = \mathtt{bv_{idx}}} \iff {c_k~{e}
    \ceeval{} \mathtt{bv_{idx}}}\) %
  then %
  \({M'(\lproj{\rel{\nu'}}) = \mathtt{bv_{val}}} \iff {c_k~{\load e}
    \ceeval{} \mathtt{bv_{val}}}\).

  \noindent
  First, we can simplify \(M'(\lproj{\rel{\nu'}})\) as
  \begin{align*}
    M'(\lproj{\rel{\nu'}})
    &= M'(select( \lproj{\smem_k},\lproj{\rel{\iota'}})) \text{ by symbolic rule
      \rulename{load}}\\
    &= M(\lproj{\smem})[M(\lproj{\rel{\iota}}) \mapsto M(\lproj{\rel{\nu}})][M' (\lproj{\rel{\iota'}})] \text{ by def.\ of \(M'\)}\\
  \end{align*}

  \noindent
  From this point, there are two cases:
  \begin{itemize}
  \item[-] The address of the load is the same as the address of the
    previous store:
    \(M' (\lproj{\rel{\iota'}}) = M(\lproj{\rel{\iota}})\),
    therefore \({M'(\lproj{\rel{\nu'}}) = M(\lproj{\rel{\nu}})}\).

    From the induction hypothesis, the concrete index of the load
    evaluates to \(M'(\lproj{\rel{\iota'}})\), i.e.
    \(c_k~{e} \ceeval{} M'( \lproj{\rel{\iota'}})\) which can be
    rewritten as \(c_k~{e} \ceeval{} M(\lproj{\rel{\iota}})\).
    From concrete rule \rulename{store} and
    \(c_{k-1} \concsym{l}{M} s_{k-1}\), we know that the concrete
    memory from \(c_{k-1}\) to \(c_k\) is updated at index
    \(M(\lproj{\rel{\iota}})\) to map to the value
    \(M(\lproj{\rel{\nu}})\), which leads to
    \(c_k \load e \ceeval{} M(\lproj{\rel{\nu}})\) and, by rewriting,
    to \({c_k~{\load e} \ceeval{} M'(\lproj{\rel{\nu'}})}\).
    Therefore we have shown that
    \({M'(\lproj{\rel{\nu'}}) = \mathtt{bv_{val}}} \iff {c_k~{\load e}
      \ceeval{} \mathtt{bv_{val}}}\)

  \item[-] The address of the load is different from the address of
    the previous store:
    \(M' (\lproj{\rel{\iota'}}) \neq M(\lproj{\rel{\iota}})\), therefore
    \({M' (\lproj{\rel{\nu'}}) = M(\lproj{\smem})[M' (\lproj{\rel{\iota'}})]}\).

    From the induction hypothesis, the concrete index of the load
    evaluates to \(M'(\lproj{\rel{\iota'}})\), i.e.
    \(c_k~{e} \ceeval{} M'(\lproj{\rel{\iota'}})\).
    From concrete rule \rulename{store}, we know that the concrete
    memory from \(c_{k-1}\) to \(c_k\) is only updated at address
    \(M(\lproj{\rel{\iota}})\) and untouched at address
    \(M'(\lproj{\rel{\iota'}})\). %
    Plus, we know from \(c_{k-1} \concsym{l}{M} s_{k-1}\) that address
    \(M'(\lproj{\rel{\iota'}})\) maps to
    \(M(\lproj{\smem})[M'(\lproj{\rel{\iota'}})]\) in \(c_{k-1}\).
    Therefore, in configuration \(c_k\), index
    \(M'(\lproj{\rel{\iota'}})\) maps to
    \(M(\lproj{\smem})[M'(\lproj{\rel{\iota'}})]\) which, by rewriting,
    leads to \({c_k~{\load e} \ceeval{} M'(\lproj{\rel{\nu'}})}\).
    Therefore we have shown that
    \({M'(\lproj{\rel{\nu'}}) = \mathtt{bv_{val}}} \iff {c_k~{\load e}
      \ceeval{} \mathtt{bv_{val}}}\).
  \end{itemize}

  \textbf{Other cases}: The reasoning is analogous. For the
  non-deterministic rules \rulename{ite-true}, \rulename{ite-false},
  and \rulename{d-jump}, because the leakages \(t\) and \(t'\)
  determine the control flow of the program, there exist a
  \emph{unique symbolic rule} that can be applied to match the
  execution of both \(c_{k-1}\) and \(c_{k-1}'\).
\end{proof}

\subsubsection{Proof of Correctness of RelSE
  (\Cref{thm:correctness})}\label{app:correctness}
\correctness*

\begin{proof} \emph{(Induction on k).}  Case \(k = 0\) is trivial. %
  
  Let us consider a symbolic configuration
  \(s_{k-1} \mydef \iconfold{\locvar}{\regmap}{\smem}{\pc}\) for which
  the induction hypothesis holds. Formally, for each model \(M\) and
  configurations \(c_0\), \(c_{k-1}\) such that
  \(c_0 \concsym{p}{M} s_0\) and \(c_{k-1}\concsym{p}{M} s_{k-1}\),
  we have \(c_0 \cleval{}^{k-1} c_{k-1}\).
  We need to show for each symbolic step \(s_{k-1} \ieval{} s_{k}\),
  that for each model \(M'\) and configurations \(c_0\) and \(c_k\)
  such that \(c_0 \concsym{p}{M'} s_0\) and
  \(c_{k}\concsym{p}{M'} s_{k}\), we have
  \(c_{0} \cleval{}^{k} c_k\).

Let \(\pc'\) be the new path predicate in configuration \(s_k\).  Note
that because \(\pc\) is a sub-formula of \(\pc'\), we also have
\(M' \sat \pc\). Therefore, there exists \(c_{k-1}\) such that
\(c_{k-1}\concsym{p}{M'} s_{k-1}\), and because \(\concsym{p}{M'}\) is
a tight relation, \(c_{k-1}\) is unique.  Additionally, from the
induction hypothesis for all concrete configuration \(c_0\) such that
\(c_0 \concsym{p}{M'} s_0\), we have \(c_0 \cleval{}^{k-1} c_{k-1}\).

Finally, because the symbolic execution is updated \emph{without
  over-approximation}, we also have \(c_{k-1} \cleval{} c_k\).
Therefore, for each model \(M'\) and configuration \(c_0\) and \(c_k\)
such that \(c_0 \concsym{p}{M'} s_0\) and
\(c_{k} \concsym{p}{M'} s_{k}\), we have \(c_{0} \cleval{}^{k} c_k\).
\end{proof}

\subsubsection{Proof of CT Security
  (\Cref{thm:bv})}\label{app:ct-security}
\security*

\begin{proof} \emph{(Induction on k).} Let $s_0$ be an initial
  symbolic configuration for which the symbolic evaluation never gets
  stuck. Let us consider a model \(M\) and concrete configurations
  \(\cconfvar_0 \concsym{l}{M} s_0\),
  \(\cconfvar'_0 \concsym{r}{M} s_0\), for which the induction
  hypothesis holds at step \(k\), meaning that for all
  \(c_{k} \mydef \cconf{\locvar}{\cregmap}{\cmem}\) and
  \(c_{k'} \mydef \cconf{\locvar'}{\cregmap'}{\cmem'}\) such that
  $c_0 \cleval{\leakvar}^k c_k$, $c_0' \cleval{\leakvar'}^k c_k'$,
  then \(\leakvar = \leakvar'\). We show that \cref{thm:bv} still
  holds at step \(k+1\).

  From \cref{thm:completeness}, there exists
  \(s_k \mydef \iconfold{\locvar_s}{\regmap}{\smem}{\pc}\) such that:
  \begin{equation}
    \label{eq:2}
    s_0 \ieval{}^k s_k ~\wedge~ c_k \concsym{l}{M} s_k ~\wedge~ c_k' \concsym{r}{M} s_k
  \end{equation}
  Note that from \cref{eq:2} and \cref{def:concsym}, we have
  \(l_s = l = l'\), therefore the same instructions and expression are
  evaluated in configurations \(c_k, c_k', \text{ and } s_k\).

  Because the symbolic execution does not get stuck, there exists
  \(s_{k+1}\) such that $s_k \ieval{} s_{k+1}$. We show by
  contradiction that the leakages \(\mathtt{bv}\) and \(\mathtt{bv'}\)
  produced by \(c_{k} \cleval{\mathtt{bv}} c_{k+1}\) and
  \(c_{k}' \cleval{\mathtt{bv'}} c_{k+1}'\) are necessarily equal.

  Suppose that \(c_k\) and \(c_k'\) produce distinct leakages. %
  This can happen during the evaluation of a rule \rulename{load},
  \rulename{d\_jump}, \rulename{ite}, \rulename{store}.  
  
  \textbf{Case \rulename{load}}: The evaluation of the expression
  \(\load e\) in configurations \(c_k\) and \(c_k'\) produces
  leakages \(\mathtt{bv}\) and \(\mathtt{bv'}\) and, assuming the load
  is insecure, we have \(\mathtt{bv} \neq \mathtt{bv'}\).

  Let \(\rel{\varphi}\) be the evaluation of the leakage in the
  symbolic configuration:
  \(\econf{\regmap}{\smem}{e} \eeval{} \rel{\varphi}\). %
  From \cref{eq:2}, \cref{def:concsym} and because symbolic execution
  does not get stuck, there exists \(M\) s.t. \(M \sat \pc\),
  \(\mathtt{bv} = M(\lproj{\rel{\varphi}})\) and
  \(\mathtt{bv'} = M(\rproj{\rel{\varphi}})\).

  Because we assumed \(\mathtt{bv} \neq \mathtt{bv'}\), then
  \(M(\lproj{\rel{\varphi}}) \neq M(\rproj{\rel{\varphi}})\).
  Therefore, we have
  \(M \sat \pc \wedge \lproj{\rel{\varphi}} \neq
  \rproj{\rel{\varphi}}\), meaning that \(\secleak\) evaluates to
  false and the symbolic execution is stuck.
  However, because \(s_k\) is non-blocking we have a contradiction.
  Therefore \(\mathtt{bv} = \mathtt{bv'}\).
  
  \textbf{Cases \rulename{d\_jump}, \rulename{ite}, \rulename{store}}:
  The reasoning is analogous.

  \smallskip
  \noindent We have shown that the hypothesis holds for \(k+1\). If
  $s_0 \ieval{}^{k+1} s_{k+1}$, then for all model \(M\) and
  low-equivalent initial configurations
  \(\cconfvar_0 \concsym{l}{M} s_0 \) and
  \(\cconfvar'_0 \concsym{r}{M} s_0\) such that %
  \(\cconfvar_0 \cleval{\leakvar}^k \cconfvar_k \cleval{\mathtt{bv}}
  \cconfvar_{k+1}\) %
  and %
  \(\cconfvar_0' \cleval{\leakvar'}^k \cconfvar'_k
  \cleval{\mathtt{bv'}} \cconfvar'_{k+1}\) where \(t = t'\), %
  then
  \(\leakvar \cdot [\mathtt{bv}] = \leakvar' \cdot [\mathtt{bv'}]\).
\end{proof}

\subsubsection{Proof of Bug-Finding for CT
  (\Cref{thm:bf})}\label{app:ct-bf}
\bugfinding*

\begin{proof} Let us consider symbolic configurations \(s_0\) and
  \(s_k\) such that \(s_0 \ieval{}^k s_k\) and \(s_k\) is stuck. This
  can happen during the evaluation of a rule \rulename{load},
  \rulename{d\_jump}, \rulename{ite}, \rulename{store}.

  \textbf{Case \rulename{load}}: where an expression
  \(\load e\) in the configuration
  \(s_k \mydef \iconfold{\locvar_s}{\regmap}{\smem}{\pc}\) produces a
  leakage \(\rel{\varphi}\) st.
  \(\econf{\regmap}{\smem}{e} \eeval{} \rel{\varphi}\).
  
  This evaluation blocks iff \(\neg \secleak(\rel{\varphi})\), meaning
  that there exists a model \(M\) such that
  \begin{equation}
    \label{eq:3}
    M \sat \pc \wedge (\lproj{\rel{\varphi}} \neq
    \rproj{\rel{\varphi}})
  \end{equation}

  Let us consider the concrete configurations \(c_0\), \(c_0'\),
  \({c_{k} \mydef \cconf{\locvar}{\cregmap}{\cmem}}\), and
  \({c_{k'} \mydef \cconf{\locvar'}{\cregmap'}{\cmem'}}\) such that:
  \[c_0 \concsym{l}{M} s_0, c_0' \concsym{r}{M} s_0, c_k
    \concsym{l}{M} s_k \text{ and } c_k' \concsym{r}{M} s_k\] %
  It follows by \cref{thm:correctness}, that
  \(c_0 \cleval{}^k c_k \text{ and } c_0' \cleval{}^k c_k'\).

  From \cref{def:concsym}, because
  \(c_0 \concsym{l}{M} s_0 \text{ and } c_0' \concsym{r}{M} s_0\) we
  have \(c_0 \loweq c_0'\) and because
  \(c_k \concsym{l}{M} s_k \text{ and } c_k' \concsym{r}{M} s_k\), we
  have \(l_s = l = l'\).

  Therefore the evaluation of \(c_k\) and \(c_k'\) also contains the
  expression \(\load e\), producing leakages \(\mathtt{bv}\) and
  \(\mathtt{bv'}\) st.  \(c_k~\load e \ceeval{} \mathtt{bv}\)
  and \(c_k'~\load e \ceeval{} \mathtt{bv'}\).

  From \cref{def:concsym} we have
  \(\mathtt{bv} = M(\lproj{\rel{\varphi}})\) and
  \(\mathtt{bv'} = M(\rproj{\rel{\varphi}})\), and from \cref{eq:3},
  we can deduce \(\mathtt{bv}\neq\mathtt{bv'}\).
  
  \smallskip
  \noindent 
  Therefore, we have a model \(M\) and concrete configurations
  \(\cconfvar_0 \concsym{l}{M} s_0\), %
  \(\cconfvar_0' \concsym{r}{M} s_0 \), %
  \(\cconfvar_k \concsym{l}{M} s_k \) and %
  \(\cconfvar_k' \concsym{r}{M} s_k\) such that,
  \begin{multline*}
    \cconfvar_0 \loweq \cconfvar_0'\ ~\wedge~ %
    \cconfvar_0 \cleval{\leakvar}^k \cconfvar_k \cleval{\mathtt{bv}} \cconfvar_{k+1} ~\wedge~ %
    \cconfvar_0' \cleval{\leakvar}^k \cconfvar'_k \cleval{\mathtt{bv'}} \cconfvar'_{k+1} \\ %
    \wedge \leakvar \cdot [\mathtt{bv}] \neq \leakvar' \cdot [\mathtt{bv'}]
  \end{multline*}
  which shows that the program is insecure.
  
  \textbf{Cases \rulename{d\_jump}, \rulename{ite}, \rulename{store}}:
  The reasoning is analogous.
\end{proof}

\subsection{Usability: Stubs for Input Specification}\label{app:stubs}
We enable the specification of high and low variables in C source code
using dummy functions that are stubbed in the symbolic execution
(cf.~\cref{ex:stub}).  Note that this is at the cost of portability
(it can only be used around C static libraries or when in possession
of the C source code). If portability is required, the user can still
refer to the binary-level specification method (\cref{sec:std_relse})
which relies on manual reverse-engineering to find the offsets of
secrets relatively to the initial \texttt{esp}.

A call to a function \codeinline{high_input_$n$($addr$)}, where \(n\)
is a constant value, specifies that the memory must be initialized with
\(n\) secret bytes, starting at address \(addr\).

\begin{example}[Stub for specifying high locations]\label{ex:stub}
  A user can write a wrapper around a function \texttt{foo} to mark
  its arguments as low or high as shown in~\Cref{list:wrapper}. The
  function \texttt{foo} can be defined in an external library which
  must be statically linked with the wrapper program.
  \begin{center}
  \begin{minipage}{.9\linewidth}
\begin{ccode}[caption={Wrapper around external function \texttt{foo}},label={list:wrapper}] %
uint8_t secret[4]; uint8_t public[4];
high_input_4(secret); //4 bytes high input
low_input_4(public);  //4 bytes low input
return foo(secret, public);
\end{ccode}
\end{minipage}
\end{center}
 During the symbolic execution, the function %
  \codeinline{high_input_4(secret)} is encountered, it is stubbed as:
  \begin{align*}
    \text{\codeinline{@[secret+0]}} := \pair{\beta_{0}}{\beta'_{0}}\quad & \quad
    \text{\codeinline{@[secret+1]}} := \pair{\beta_{1}}{\beta'_{1}}\\
    \text{\codeinline{@[secret+2]}} := \pair{\beta_{2}}{\beta'_{2}}\quad & \quad
    \text{\codeinline{@[secret+3]}} := \pair{\beta_{3}}{\beta'_{3}}
  \end{align*}
  where \(\beta_i, \beta_i'\) are fresh 8-bit bitvector variables.
\end{example}

\subsection{Zoom on the Lucky13 Attack}\label{app:lucky13}
Lucky 13~\cite{al_fardan_lucky_2013} is a famous attack exploiting
timing variations in TLS CBC-mode decryption to build a Vaudenay’s
padding oracle attack and enable plaintext
recovery~\cite{al_fardan_lucky_2013,ronen_pseudo_2018}. We do not
actually mount an attack but show how to find violations of
constant-time that could potentially be exploited to mount such
attack.

We focus on the function \texttt{tls-cbc-remove-padding} which checks
and removes the padding of the decrypted record.  We extract the
vulnerable version from
OpenSSL-1.0.1\textsuperscript{\ref{fnote:lucky13}}
and its patch from~\cite{almeida_verifying_2016}. Finally, we check
that no information is leaked during the padding check by specifying
the record data as private.

On the \emph{insecure version}, \brelse{} accurately reports 5 violations,
and for each violation, returns the address of the faulty instruction,
the execution trace which can be visualized with IDA, and an input
triggering the violation. For instance, on the portion of code in
\Cref{lst:remove_padding}, three violation are reported: two
conditional statement depending on the padding length at lines 3 and
4, and a memory access depending on the padding length at line 4. For
the conditional at line 3, when the length \cinline{LEN} of the record
data is set to 63, %
\brelse{} returns in 0.11s the counterexample %
``\cinline{data_l[62]=0; data_r[62]=16}'', %
meaning that an execution with a padding length set to 0 will take a
different path that an execution with a padding length set to 16.

\begin{center}
  \begin{minipage}{.9\linewidth}
\begin{ccode}[numbers=left,caption={Padding check in OpenSSL-1.0.1},label={lst:remove_padding}]
pad_len = rec->data[LEN-1]; // Get padding length
[...]
for (i = LEN - pad_len; i < LEN; i++)
  if (rec->data[i] != pad_len)
    return -1; // Incorrect padding
\end{ccode}    
\end{minipage}
\end{center}

On the \emph{secure version}, when the length \cinline{LEN} of the record
data is set to 63, \brelse{} explores all the paths in 400s and
reports no vulnerability.

\end{document}
